\documentclass{amsart}
\usepackage{amsfonts,amssymb,amsmath,amsthm}
\usepackage{url}
\usepackage{enumerate}
\usepackage{graphicx}
\usepackage{epstopdf}
\usepackage{float}
\usepackage{hyperref}
\usepackage{wrapfig}
\newtheorem{thrm}{Theorem}[section]

\theoremstyle{definition}

\numberwithin{equation}{section}
\author{Dejan Kovacevic}
\address{
21 Neilson Drive
Toronto, M9C 1V3 Canada\\
orcid.org/0000-0003-2214-7380}
\email{kodza@yahoo.com}
\keywords{Navier-Stokes, Clay Mathematics Institute, Incompressible Fluid}
\begin{document}
\title{Incompressible Navier-Stokes Equations:  Example of no solution at $R^3$ and $t=0$}
\maketitle
\begin{abstract}
We provide an example of a smooth, divergence-free $\nabla \cdot \vec{u}(\vec{x})=0$ velocity vector field $\vec{u}(\vec{x})$ for  incompressible fluid occupying all of $R^{3}$ space, and smooth vector field $\vec{f}(\vec{x}, t)$ for which the Navier-Stokes equation for incompressible fluid does not have a solution for any position in space $\vec{x}\in R^{3} $ at $t=0$. The velocity vector field $u_{i} (\vec{x})=2\frac{x_{h(i-1)} -x_{h(i+1)} }{\left(1+\sum _{j=1}^{3}x_{j}^{2}  \right)^{2} }$ ; $i=\{ 1,2,3\} $ where $h(l)=\left\{\begin{array}{ccc} {l} & {;1\le l\le 3} & {} \\ {1} & {;l=4} & {} \\ {3} & {;l=0} & {} \end{array}\right.$ is smooth, divergence-free, continuously differentiable $u(\vec{x})\in C^{\infty }$, has bounded energy $\int _{R^{3} }\left|\vec{u}\right|^{2} dx=\pi ^{2}$, zero velocity at coordinate origin, and velocity converges to zero for $\left|\vec{x}\right|\to \infty$. The vector field $\vec{f}(\vec{x},t)=(0,0,\frac{1}{1+t^{2} (\sum _{j=1}^{3}x_{j}  )^{2} )} $ is smooth, continuously differentiable $f(\vec{x},t)\in C^{\infty }$, converging to zero for $\left|\vec{x}\right|\to \infty$. Applying $\vec{u}(\vec{x})$ and $\vec{f}(\vec{x}, t)$ in the Navier-Stokes equation for incompressible fluid results with three mutually different solutions for pressure $p(\vec{x}, t)$, one of which includes zero division with zero $\frac{0}{0}$ term at $t=0$, which is indeterminate for all positions $\vec{x} \in R^{3}$.
\end{abstract}
\maketitle
\section{Introduction}
An important and still unresolved problem in fluid dynamics is the question of global regularity in three dimensional Euclidian $R^{3}$ space, utilizing Navier-Stokes equation for incompressible fluid:
\[\frac{\partial \vec{u}}{\partial t} +(\vec{u}\cdot \nabla )\vec{u}=-\frac{\nabla p}{\rho } +\nu \Delta \vec{u}+\vec{f}\]
for $\nabla \cdot \vec{u}=0$ at any position in space $\vec{x}\in R^{3} $ at any time $t\ge 0$.
Fluid velocity, a physical quantity representing the ratio between fluid parcel spatial position change and the increment of elapsed time, satisfying realistic boundary conditions- could conceivably develop singularity in finite time. The literature refers to such phenomenon as a "blow-up", as, over some finite time, a mathematical representation of fluid velocity and it’s derivatives, could reach values corresponding to physically unreasonable results. As per the Clay Mathematics Institute official existence and smoothness of the Navier-Stokes equation problem statement: \emph{"A fundamental problem in analysis is to decide whether such smooth, \textbf{physically reasonable} solutions exist for the Navier–Stokes equations"}. In addition, it is not known if a smooth, divergence-fee vector field $\vec{u}^{0}(\vec{x})$ and smooth $\vec{f}(\vec{x},t)$ exist, for which there exist no solutions for pressure $p(\vec{x}, t)$ and given velocity vector field $\vec{u}$ at any position in space $\vec{x}\in R^{3} $ at $t=0$.
\newline
In this paper we share specific example of the fluid velocity vector field $\vec{u}^{0}(\vec{x})=\vec{u}(\vec{x})$ for fluid occupying all of $R^{3}$ space, as well as the $\vec{f}(\vec{x},t)$ vector field, for which we prove that the Navier-Stokes equation for incompressible fluid does not have a solution at any position in space $\vec{x} \in R^{3}$ at $t=0$.
\newline
We present single theorem which encapsulates the following approach and results:
\newline
\textbf{Definition of vector fields:}
The theorem statement includes definition of fluid velocity vector field as:
\begin{wrapfigure}{r}{0.3\textwidth}
  \begin{center}
    \includegraphics[width=0.3\textwidth]{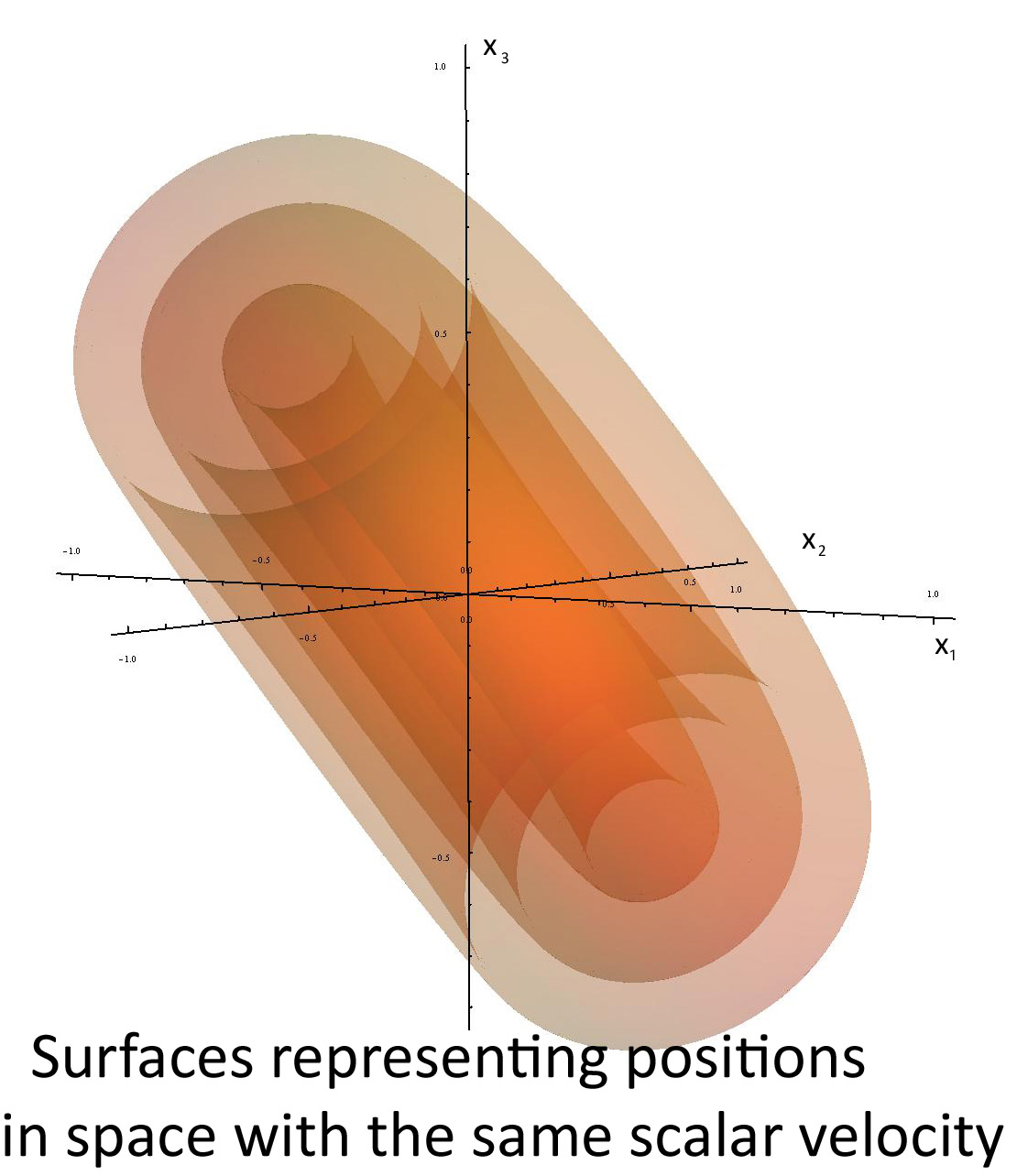}
    \label{fig:databaseUserTable}
  \end{center}
\end{wrapfigure}
\[u_{i} (\vec{x})=2\frac{x_{h(i-1)} -x_{h(i+1)} }{\left(1+\sum _{j=1}^{3}x_{j}^{2}  \right)^{2} } \]
for $i=\{ 1,2,3\} $ where $h(l)$ is defined as:
\[h(l)=\left\{\begin{array}{ccc} {l} & {;1\le l\le 3} & {} \\ {1} & {;l=4} & {} \\ {3} & {;l=0} & {} \end{array}\right. \]
for any $\vec{x}\in R^{3}$.
\newline
Also, we define external force related vector field as:
\[\vec{f}(\vec{x},t)=(0,0,\frac{1}{1+t^{2} (\sum _{j=1}^{3}x_{j}  )^{2}})\]
for any $\vec{x}\in R^{3}$ and $t\ge 0$.
\newline
\textbf{Analysis of $\vec{u}(\vec{x})$ vector field:}
Starting from the $\vec{u}(\vec{x})$ definition, we analyze the fluid velocity vector field:
\begin{itemize}
    \item Proving that the velocity vector field $\vec{u}(\vec{x})$ is divergence free $\nabla\cdot\vec{u}=\vec{0}$.
    \item Integrating square of fluid scalar velocity $\left|\vec{u}\right|^{2}$ over the whole $R^{3}$ space in order to verify energy as bounded and finite. The result of integration is $\int _{R^{3} }\left|\vec{u}\right|^{2} dx=\pi ^{2}$, which is constant.
    \item Proving that velocity vector field $\vec{u}(\vec{x})$ is continuously differentiable $\vec{u}(\vec{x})\in C^{\infty }$.
    \item Deriving the general form of the partial derivative $\frac{\partial^{\alpha} \vec{u}}{\partial x_{k}^{\alpha}}$ of any order $\alpha$, $k=\{1,2,3\}$.
    \item Proving that the partial derivative $\frac{\partial^{\alpha} \vec{u}}{\partial x_{k}^{\alpha}}$ of any order $\alpha$, $k=\{1,2,3\}$, is zero vector at coordinate origin $\vec{x}=\vec{0}$ and that it converges to zero vector when $\left|\vec{x}\right| \to \infty$.
    \item Proving that partial derivative $\frac{\partial^{\alpha} \vec{u}}{\partial x_{k}^{\alpha}}$ of any order $\alpha$, $k=\{1,2,3\}$ must be finite $\left|\frac{\partial^{\alpha} \vec{u}}{\partial x_{k}^{\alpha}}\right|\leq C_{\alpha}$ for $C_{\alpha} \in R$, for any position in space $\vec{x}\in R^{3}$.
\end{itemize}
\begin{wrapfigure}{r}{0.3\textwidth}
  \begin{center}
    \includegraphics[width=0.3\textwidth]{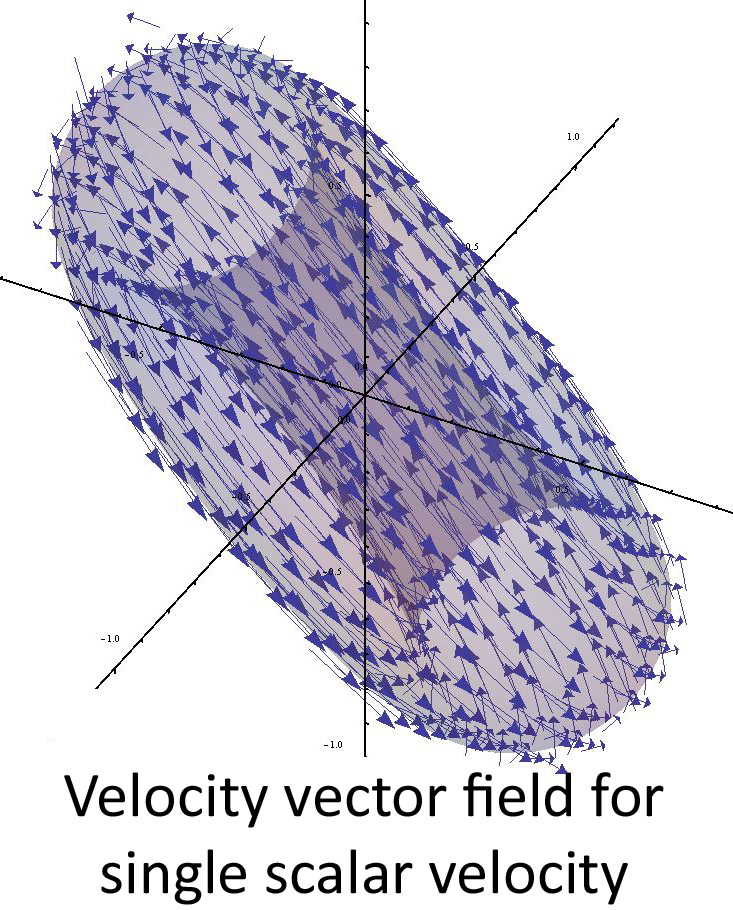}
    \label{fig:databaseUserTable2}
  \end{center}
\end{wrapfigure}
\textbf{Analysis of $\vec{f}(\vec{x},t)$ vector field:}
We perform an analysis for the external force vector field $\vec{f}(\vec{x},t)$:
\begin{itemize}
    \item Proving that vector field $\vec{f}(\vec{x},t)$ is continuously differentiable $\vec{f}(\vec{x})\in C^{\infty }$.
    \item Deriving $ \deg _{x} (\frac{\partial ^{\alpha } }{\partial x^{\alpha } } \frac{\partial ^{m} }{\partial t^{m} }\vec{f})$ of any order $\alpha$, $m$
      \item Deriving $ \deg _{t} (\frac{\partial ^{\alpha } }{\partial x^{\alpha } } \frac{\partial ^{m} }{\partial t^{m} }\vec{f})$ of any order $\alpha$, $m$
    \item Proving that the partial derivative $\frac{\partial ^{\alpha } }{\partial x^{\alpha } } \frac{\partial ^{m} }{\partial t^{m} }\vec{f}$ of any order $\alpha$, $m$ is zero vector at the coordinate origin $\vec{x}=\vec{0}$ and that it converges to zero vector when $\left|\vec{x}\right| \to \infty$.
\end{itemize}
\textbf{Appllication of $\vec{u}(\vec{x})$ and  $\vec{f}(\vec{x},t)$ in the Navier-Stokes equation and solving for pressure $p(\vec{x}, t)$:}
We apply the velocity vector field $\vec{u}(\vec{x})$ and $\vec{f}(\vec{x},t)$ in the Navier-Stokes equation for incompressible fluid, obtaining the following results:
\begin{itemize}
    \item We apply $\vec{u}(\vec{x})$ and $\vec{f}(\vec{x},t)$ to each applicable term of the Navier-Stokes equation, and we re-arrange the equation so that $\frac{\nabla p}{\rho}$ is the only term on one side of the resulting equation.
    \item As $\nabla p=(\frac{\partial}{\partial x_{1}},\frac{\partial}{\partial x_{2}},\frac{\partial}{\partial x_{3}})p$, we integrate the resulting vector field components for $x_{1}$,$x_{2}$,$x_{3}$ expecting that the resulting pressure $p(\vec{x}, t)$ must be the same for all three integrations performed.
    \item Once the three integrations are performed for $x_{1}$,$x_{2}$,$x_{3}$, we obtain the three mutually different results for pressure $p(\vec{x}, t)$ for all positions in space $\vec{x}\in R^{3}$ and for $t\ge 0$. In addition, one of the resulting pressure solutions includes the term $\frac{ArcTan(t (x_{1}+x_{2}+x_{3}))}{t}$ which at $t=0$ evaluates to $\frac{0}{0}$, which is indeterminate for all positions in space $\vec{x}\in R^{3}$ at $t=0$.
\end{itemize}
\textbf{Conclusion:} Based on the analysis performed and the obtained results, we conclude that for fluid velocity $\vec{u}(\vec{x})$ and $\vec{f}(\vec{x},t)$ vector fields as specified, the Navier-Stokes equation for incompressible fluid does not have a solution for all positions in space $\vec{x}\in R^{3}$ at $t=0$.
\begin{thrm}\label{T0}
\noindent Let $\vec{u}(\vec{x})=(u_{i} (\vec{x}))_{1\le i\le 3} \in R^{3}$ be divergence-free $\nabla \cdot \vec{u}=0$ velocity vector field for incompressible fluid $\rho =const$ occupying all of $R^{3}$ space, where $p(\vec{x},t)\in R$ represents fluid pressure and $\vec{f}(\vec{x},t)$ represents the vector field related to the external force applied on fluid, for any position $\vec{x}\in R^{3} $ in space at $t\ge 0$.
\newline Than the Navier-Stokes equation for incompressible fluid
\[\frac{\partial \vec{u}}{\partial t} +(\vec{u}\cdot \nabla )\vec{u}=-\frac{\nabla p}{\rho } +\nu \Delta \vec{u}+\vec{f}\]
\textbf{does not have solution} for all positions in space $\vec{x}\in R^{3} $ at $t=0$, for vector fields $\vec{u}(\vec{x})$ and $\vec{f}(\vec{x},t)$ defined as:
\[u_{i} (\vec{x})=2\frac{x_{h(i-1)} -x_{h(i+1)} }{\left(1+\sum _{j=1}^{3}x_{j}^{2}  \right)^{2}}\quad ;i=\{1,2,3\}\]
where $h(l)$ is defined as
\[h(l)=\left\{\begin{array}{ccc} {l} & {;1\le l\le 3} & {} \\ {1} & {;l=4} & {} \\ {3} & {;l=0} & {} \end{array}\right.\]
\noindent and vector field $\vec{f}(\vec{x},t)$ defined as
\[\vec{f}(\vec{x},t)=(0,0,\frac{1}{(1+t^{2} (\sum _{j=1}^{3}x_{j}  )^{2} )} )\]
\end{thrm}

\begin{proof}
In order to simplify the manipulation of equations, let us define
\begin{align}
    d_{i} =x_{h(i-1)} -x_{h(i+1)}
\label{eq1.1}\end{align}
for $i=\{1,2,3\}$ and
\begin{align}
    S=1+\sum _{j=1}^{3}x_{j}^{2}  =1+\left|\vec{x}\right|^{2}
\label{eq1.2}\end{align}
for any $\vec{x}\in R^{3}$. Once (\ref{eq1.1}) and (\ref{eq1.2}) are applied to the statement for fluid velocity vector field $u_{i}(\vec{x})$ as defined by the theorem statement:
\begin{align}
    u_{i} (\vec{x})=2\frac{x_{h(i-1)} -x_{h(i+1)} }{\left(1+\sum _{j=1}^{3}x_{j}^{2}  \right)^{2} }
\label{eq1.3}\end{align}
for any $\vec{x}\in R^{3};i=\{1,2,3\}$, vector field components can be represented as:
\begin{align}
   u_{i} (\vec{x})=2\frac{d_{i} }{S^{2} }
\label{eq1.4}\end{align}
for any $\vec{x}\in R^{3};i=\{1,2,3\}$.
\newline\newline
\textbf{Validation of $\nabla \cdot \vec{u}(\vec{x})=0$ }
\newline
Let us verify the velocity vector field $\vec{u}(\vec{x})$ as divergence-free. The divergence of the velocity vector field can be expressed as
\begin{align}
  \nabla \cdot \vec{u}=\sum _{i=1}^{3}\frac{\partial u_{i} }{\partial x_{i} }
\label{eq1.5}\end{align}
for any $\vec{x}\in R^{3};i=\{1,2,3\}$. Applying (\ref{eq1.4}) in (\ref{eq1.5}):
\begin{align}
  \nabla \cdot \vec{u}=\sum _{i=1}^{3}\frac{\partial }{\partial x_{i} } \left(2\frac{d_{i} }{S^{2} } \right)
\label{eq1.6}\end{align}
\begin{align}
  \nabla \cdot \vec{u}=\sum _{i=1}^{3}\frac{2}{S^{4} } \left(\frac{\partial d_{i} }{\partial x_{i} } S^{2} -d_{i} 2S\frac{\partial S}{\partial x_{i} } \right)
\label{eq1.7}\end{align}
for any $\vec{x}\in R^{3};i=\{1,2,3\}$. Let us analyze $\frac{\partial d_{i} }{\partial x_{i} }$ within the brackets of the statement (\ref{eq1.7}):
\begin{align}
  \frac{\partial d_{i} }{\partial x_{i} } =\frac{\partial }{\partial x_{i} } (x_{h(i-1)} -x_{h(i+1)} )
\label{eq1.8}\end{align}
for any $\vec{x}\in R^{3};i=\{1,2,3\}$.
As the indices $h(i-1)\ne i$ and $h(i+1)\ne i$ in partial differentiation in statement (\ref{eq1.8}) are always different compared to the index $i$, then both partial differentiations must be zero:
\begin{align}
  \frac{\partial x_{h(i-1)} }{\partial x_{i} } =\frac{\partial x_{h(i+1)} }{\partial x_{i} } =0
\label{eq1.9}\end{align}
for $i=\{1,2,3\}$.
Applying (\ref{eq1.9}) in (\ref{eq1.8}):
\begin{align}
  \frac{\partial d_{i} }{\partial x_{i} } =\frac{\partial }{\partial x_{i} } (x_{h(i-1)} -x_{h(i+1)} )=0
\label{eq1.10}\end{align}
for $i=\{1,2,3\}$.
Applying (\ref{eq1.10}) in (\ref{eq1.7}):
\begin{align}
  \nabla \cdot \vec{u}=-\sum _{i=1}^{3}\frac{4d_{i} }{S^{3} } \frac{\partial S}{\partial x_{i} }
\label{eq1.12}\end{align}
for any $\vec{x}\in R^{3};i=\{1,2,3\}$.
As per the definition of $S$ by statement (\ref{eq1.2}), let us derive $\frac{\partial S}{\partial x_{i}}$:
\begin{align}
    \frac{\partial S}{\partial x_{i} } =\frac{\partial }{\partial x_{i} } \left(1+\sum _{j=1}^{3}x_{j}^{2}  \right)=\frac{\partial }{\partial x_{i} } \left(\sum _{j=1}^{3}x_{j}^{2}  \right)=\sum _{j=1}^{3}\frac{\partial x_{j}^{2} }{\partial x_{i} }  =2x_{i}
\label{eq1.13}\end{align}
for $i=\{1,2,3\}$.
Applying (\ref{eq1.13}) in (\ref{eq1.12}):
\begin{align}
    \nabla \cdot \vec{u}=-\frac{8}{S^{3} } \sum _{i=1}^{3}d_{i} x_{i}
\label{eq1.15}\end{align}
for any $\vec{x}\in R^{3};i=\{1,2,3\}$.
Applying (\ref{eq1.1}) in (\ref{eq1.15}) and following expanding:
\[\nabla \cdot \vec{u}=-\frac{8}{S^{3} } (2(x_{3} -x_{2} )x_{1} +2(x_{1} -x_{3} )x_{2} +2(x_{2} -x_{1} )x_{3} )\]
\[\nabla \cdot \vec{u}=-\frac{8}{S^{3} } (2x_{3} x_{1} -2x_{2} x_{1} +2x_{1} x_{2} -2x_{3} x_{2} +2x_{2} x_{3} -2x_{1} x_{3} )=0\]
\begin{align}
   \nabla \cdot \vec{u}=0
\label{eq1.19}\end{align}
for any $\vec{x}\in R^{3}$,
which proves that $\vec{u}(\vec{x})$ is divergence-free for any $\vec{x}\in R^{3}$.
\newline\newline
\textbf{Validation of $\vec{u}(\vec{x})$ continuity, convergence to zero at infinity, zero at coordinate origin}
\newline
Based on the definition of the fluid velocity vector field $\vec{u}(\vec{x})$, as per the theorem statement, we can conclude that vector field $\vec{u}(\vec{x})$ does not have singularity at any position in space $\vec{x}\in R^{3}$. The denominator for each of the velocity vector field components $\left\{u_{i} \right\}_{1\le i\le 3}$
\begin{align}
    \left(1+\sum _{j=1}^{3}x_{j}^{2}  \right)^{2} \ge 1
\label{eq1.19.1}\end{align}
for $i=\{1,2,3\}$ is positive definite for any position $\vec{x}\in R^{3}$.
The vector field $\vec{u}(\vec{x})$ at infinity converges to zero vector as:
\[\deg_{x} ( \vec{u}(\vec{x}))=\deg_{x}(\frac{ x_{h(i-1)} -x_{h(i+1)} }{\left(1+\sum _{j=1}^{3}x_{j}^{2}  \right)^{2}}) =1-4=-3\]
\begin{align}
    \lim _{\left|\vec{x}\right|\to \infty } \vec{u}(\vec{x})=\vec{0}
\label{eq1.29}\end{align}
Also, at the coordinate origin the velocity vector field is zero vector:
\begin{align}
    \vec{u}(\vec{0})=\vec{0}
\label{eq1.30}\end{align}
Based on statements (\ref{eq1.29}) and (\ref{eq1.30}),  we can could be conclude that fluid velocity vector field is zero vector at the coordinate origin, and converges to zero vector when the position converges to infinity.
\newline\newline
\textbf{Validation of bounded energy for $\vec{u}(\vec{x})$ }
\newline
Let us evaluate whether the fluid velocity vector field $\vec{u}(\vec{x})$ has bounded energy by integrating the square of its scalar value $\left|\vec{u}(\vec{x})\right|^{2} $ across all of $R^{3}$ space. Square of scalar value of $\vec{u}$ can be represented as:
\begin{align}
    \left|\vec{u}(\vec{x})\right|^{2} =\sum _{i=1}^{3}u_{i}^{2}
\label{eq1.40}\end{align}
for any $\vec{x}\in R^{3};i=\{1,2,3\}$.
Once the vector field components $u_{i}$ are applied to statement (\ref{eq1.40}), as per the definition of this theorem:
\[\left|\vec{u}(\vec{x})\right|^{2} =\left(2\frac{(x_{h(1-1)} -x_{h(1+1)} )}{(1+\sum _{j=1}^{3}x_{j}^{2}  )^{2} } \right)^{2} +\left(2\frac{(x_{h(2-1)} -x_{h(2+1)} )}{(1+\sum _{j=1}^{3}x_{j}^{2}  )^{2} } \right)^{2} +\left(2\frac{(x_{h(3-1)} -x_{h(3+1)} )}{(1+\sum _{j=1}^{3}x_{j}^{2}  )^{2} } \right)^{2} \]
\[\left|\vec{u}(\vec{x})\right|^{2} =4\frac{(x_{h(0)} -x_{h(2)} )^{2} }{(1+\sum _{j=1}^{3}x_{j}^{2}  )^{4} } +4\frac{(x_{h(1)} -x_{h(3)} )^{2} }{(1+\sum _{j=1}^{3}x_{j}^{2}  )^{4} } +4\frac{(x_{h(2)} -x_{h(4)} )^{2} }{(1+\sum _{j=1}^{3}x_{j}^{2}  )^{4} } \]
\[\left|\vec{u}(\vec{x})\right|^{2} =4\frac{(x_{3} -x_{2} )^{2} +(x_{1} -x_{3} )^{2} +(x_{2} -x_{1} )^{2} }{(1+\sum _{j=1}^{3}x_{j}^{2}  )^{4} } \]
\begin{align}
    \left|\vec{u}(\vec{x})\right|^{2} =8\frac{x_{1}^{2} -x_{1} x_{2} +x_{2}^{2} -x_{2} x_{3} +x_{3}^{2} -x_{3} x_{1} }{(1+\sum _{j=1}^{3}x_{j}^{2}  )^{4} }
\label{eq1.41.1}\end{align}
for any $\vec{x}\in R^{3}$.
Let us integrate the square of the scalar function $\left|\vec{u}(\vec{x})\right|^{2} $ across all of $R^{3}$ space:
\begin{align}
    \int _{R^{3} }\left|\vec{u}(\vec{x})\right|^{2} dx =\int _{R^{3} }8\frac{x_{1}^{2} -x_{1} x_{2} +x_{2}^{2} -x_{2} x_{3} +x_{3}^{2} -x_{3} x_{1} }{(1+\sum _{j=1}^{3}x_{j}^{2}  )^{4} } dx
\label{eq1.42}\end{align}
 Once integration in statement (\ref{eq1.42}) is performed, result is constant:
\begin{align}
    \int _{R^{3} }\left|\vec{u}\right|^{2} dx =\pi ^{2}
\label{eq1.43}\end{align}
As per (\ref{eq1.43}), the result of integration is constant $\pi ^{2}=Const$. Based on this, we can conclude that the fluid total kinetic energy, is constant across all of $R^{3}$ space.
\newline\newline
\textbf{Validation that vector field $\vec{u}(\vec{x})$  is continuously differentiable}
\newline
In this section, we confirm that the fluid velocity vector field $\vec{u}(\vec{x})$ is continuously differentiable. Let us apply the partial derivative $\frac{\partial}{\partial x_{k}}$ to the fluid velocity vector field $\vec{u}(\vec{x})$ components $u_{i}$:
\begin{align}
    \frac{\partial u_{i} (\vec{x})}{\partial x_{k} } =\lim _{dx_{k} \to 0} \frac{u_{i} (\vec{x}+dx_{k} )-u_{i} (\vec{x})}{dx_{k} }
\label{eq1.43.1}\end{align}
for any $\vec{x}\in R^{3}$, $k=\{1,2,3\}$, $i=\{1,2,3\}$.
Let us introduce the function $g(j,k) \in \{0,1\}$, such that it has value 1 in case that the index $j\in\{1,2,3\}$ is passed to it as a parameter, is equal to the index $k\in\{1,2,3\}$ of $x_{k}$; otherwise, it is 0:
\begin{align}
    g(j,k)=\left\{\begin{array}{cc} {0} & {;j\ne k} \\ {1} & {;j=k} \end{array}\right.
\label{eq1.43.3}\end{align}
for $k\in\{1,2,3\}$; $j\in\{1,2,3\}$.
Once the vector field components $u_{i}$, as defined by this theorem statement, are applied to the statement (\ref{eq1.43.1}) together with (\ref{eq1.43.3}), statement (\ref{eq1.43.1}) can be equivalently represented as:
\begin{align}
    \frac{\partial u_{i} (\vec{x})}{\partial x_{k} } =\lim _{dx_{k} \to 0} \frac{2\frac{(x_{h(i-1)} +g(h(i-1),k)dx_{k} )-(x_{h(i+1)} +g(h(i+1),k)dx_{k} )}{\left(1+\sum _{j=1}^{3}(x_{j} +g(j,k)dx_{k} )^{2}  \right)^{2} } -2\frac{x_{h(i-1)} -x_{h(i+1)} }{\left(1+\sum _{j=1}^{3}x_{j}^{2}  \right)^{2} } }{dx_{k} }
\label{eq1.49.5}\end{align}
for any $\vec{x}\in R^{3}$, $k=\{1,2,3\}$, $i=\{1,2,3\}$.
For terms in the brackets of statement (\ref{eq1.49.5}):
\[x_{h(i-1)} +g(h(i-1),k)dx_{k} \]
for $i=\{1,2,3\}$, $k=\{1,2,3\}$. Function $g(h(i-1),k)$ would have value 1 in case when $h(i-1)=k$, so the resulting statement in that scenario would become:
\[x_{h(i-1)} +dx_{k} \]
for $i=\{1,2,3\}$, $k=\{1,2,3\}$. Otherwise, when $h(i-1)\ne k$ function $g(j,k)$ would return zero, and the result would be
\[x_{h(i-1)} \]
for $i=\{1,2,3\}$
Once we apply the same principle for the other two pairs of terms in the second brackets of statement (\ref{eq1.49.5}):
\[(x_{h(i+1)} +g(h(i+1))dx_{k} )\]
for $i=\{1,2,3\}$, $k=\{1,2,3\}$. as well as for denominator:
\begin{align}
    \left(1+\sum _{j=1}^{3}(x_{j} +g(j,k)dx_{k} )^{2}  \right)^{2}
\label{eq1.50}\end{align}
for $k=\{1,2,3\}$,
we can conclude that the resulting statement (\ref{eq1.49.5}) will include $dx_{k}$ for a proper and matching $x_{j}$, for which indices are the same $j=k$, in order to form a correct statement for partial differentiation of $u_{i} (\vec{x})$ by $x_{k}$.
Therefore, in statement (\ref{eq1.50}) $dx_{k}$ will be included in case when $j=k$ for which $g(j,k)=1$. With that, we expand and rearrange statement (\ref{eq1.50}), which could be equivalently represented as:
\begin{align}
    \left(2x_{k} dx_{k} +dx_{k}^{2} +1+\sum _{j=1}^{3}x_{j}^{2}  \right)^{2}
\label{eq1.50.1}\end{align}
for $k=\{1,2,3\}$.
Applying (\ref{eq1.2})in (\ref{eq1.50}):
\begin{align}
    \left(S+(2x_{k} dx_{k} +dx_{k}^{2} )\right)^{2}
\label{eq1.52}\end{align}
for $k=\{1,2,3\}$.
Let us expand (\ref{eq1.52}):
\begin{align}
    S^{2} +dx_{k} (4Sx_{k} +dx_{k} (2S+4x_{k}^{2} +4x_{k} dx_{k} +dx_{k}^{2} ))
\label{eq1.53}\end{align}
Applying (\ref{eq1.53}) in (\ref{eq1.49.5})
\begin{align}
     \frac{\partial u_{i} (\vec{x})}{\partial x_{k} } =\lim _{dx_{k} \to 0} \frac{1}{dx_{k} } 2\left(\frac{(x_{h(i-1)} +g(h(i-1),k)dx_{k} )-(x_{h(i+1)} +g(h(i+1),k)dx_{k} )}{S^{2} +dx_{k} (4Sx_{k} +dx_{k} (2S+4x_{k}^{2} +4x_{k} dx_{k} +dx_{k}^{2} ))} -\frac{x_{h(i-1)} -x_{h(i+1)} }{S^{2} } \right)
\label{eq1.53.1}\end{align}
for any $\vec{x}\in R^{3}$, $k=\{1,2,3\}$, $i=\{1,2,3\}$.
After applying common denominator and simplifying statement (\ref{eq1.53.1}):
\begin{align}
    \frac{\partial u_{i} (\vec{x})}{\partial x_{k} } =\lim _{dx_{k} \to 0} \frac{2((g(h(i-1),k)-g(h(i+1),k)))S^{2}}{(S^{2} +dx_{k} (4Sx_{k} +dx_{k} (2S+4x_{k}^{2} +4x_{k} dx_{k} +dx_{k}^{2} )))S^{2} }-
\label{eq1.53.2}\end{align}
\[ -\frac{(4Sx_{k} +dx_{k} (2S+4x_{k}^{2} +4x_{k} dx_{k} +dx_{k}^{2} ))(x_{h(i-1)} -x_{h(i+1)} )}{(S^{2} +dx_{k} (4Sx_{k} +dx_{k} (2S+4x_{k}^{2} +4x_{k} dx_{k} +dx_{k}^{2} )))S^{2} }\]

for any $\vec{x}\in R^{3}$, $k=\{1,2,3\}$, $i=\{1,2,3\}$.
Once limit for $dx_{k} \to 0$ is applied:
\begin{align}
    \frac{\partial u_{i} (\vec{x})}{\partial x_{k} } =\frac{2((g(h(i-1),k)-g(h(i+1),k)))S^{2} -8Sx_{k} (x_{h(i-1)} -x_{h(i+1)} )}{S^{4} }
\label{eq1.53.3}\end{align}
for any $\vec{x}\in R^{3}$, $k=\{1,2,3\}$, $i=\{1,2,3\}$.
Depending on indices $i$ and $k$, there are three scenarios:
\newline
1)For $k=i$ then $g(h(i-1),k)=0$ and $g(h(i+1),k)=0$. Applying this in (\ref{eq1.53.3}):
\begin{align}
    \frac{\partial u_{i} (\vec{x})}{\partial x_{k} } =-\frac{8Sx_{k} (x_{h(i-1)} -x_{h(i+1)} )}{S^{4} }
\label{eq1.54}\end{align}
for any $\vec{x}\in R^{3}$, $k=\{1,2,3\}$, $i=\{1,2,3\}$.
\newline
2)For $k=i+1$ then $g(h(i-1),k)=0$ and $g(h(i+1),k)=1$. Applying this in (\ref{eq1.53.3}):
\begin{align}
    \frac{\partial u_{i} (\vec{x})}{\partial x_{k} } =-\frac{2S+8x_{k} (x_{h(i-1)} -x_{h(i+1)} )}{S^{3} }
\label{eq1.55}\end{align}
3)For $k=i-1$ then $g(h(i-1),k)=1$ and $g(h(i+1),k)=0$. Applying this in (\ref{eq1.53.3}):
\begin{align}
    \frac{\partial u_{i} (\vec{x})}{\partial x_{k} } =\frac{2S-8x_{k} (x_{h(i-1)} -x_{h(i+1)} )}{S^{3} }
\label{eq1.56}\end{align}
for any $\vec{x}\in R^{3}$, $k=\{1,2,3\}$, $i=\{1,2,3\}$.
In all three cases above (\ref{eq1.54}), (\ref{eq1.55}), (\ref{eq1.56}), the denominators are in general form $S^{n} =\left(1+\sum _{j=1}^{3}x_{j}^{2}  \right)^{n} $ for $n \in \{3,4\}$, which is positive definite and must be greater or equal to 1
\begin{align}
    \left(1+\sum _{j=1}^{3}x_{j}^{2}  \right)^{n} \ge 1
\label{eq1.56.1}\end{align}
for $n \in \{3,4\}$.
Statement (\ref{eq1.56.1}) has minimal value at the coordinate origin
\begin{align}
    \left. \left(1+\sum _{j=1}^{3}x_{j}^{2}  \right)^{n} \right|_{\vec{x}=\vec{0}} =1
\label{eq1.56.2}\end{align}
Based on this, we can conclude that the velocity vector field $\vec{u}(\vec{x})$ is continuously differentiable for first partial derivatives.
\noindent As the denominator of the first partial derivative $\frac{\partial u_{i} (\vec{x})}{\partial x_{k}}$ is in form $S^{n} =\left(1+\sum _{j=1}^{3}x_{j}^{2}  \right)^{n} $ for $n \in \{3,2\}$ , which has same form of denominator as the velocity vector field itself $S^{2} =\left(1+\sum _{j=1}^{3}x_{j}^{2}  \right)^{2}$. Repeating the same process of partial differentiation as above, we can conclude that for multiple partial derivatives $\frac{\partial ^{\alpha } u_{i} (\vec{x})}{\partial x_{k}^{\alpha } } $, the resulting derivative's denominators must have the same general form $\left(1+\sum _{j=1}^{3}x_{j}^{2}  \right)^{n} \ge 1$ ; $\{n>2\} \in N$, which must be positive definite and cannot create singularities.
\noindent In addition, partial differentiation can be repeated an unlimited number of times.
We can therefore conclude that the velocity vector field $\vec{u}(\vec{x})$ is continuously differentiable $u(\vec{x})\in C^{\infty } $ for any position $\vec{x}\in R^{3} $.
\newline\newline
\textbf{Derivation of partial derivative $\frac{\partial^{\alpha} u_{i}}{\partial x_{k}^{\alpha}}$ of any order $\alpha$}
\newline
Let us explore partial derivatives of the velocity vector field $\vec{u}(\vec{x})=(u_{i} (\vec{x}))_{1\le i\le 3} \in R^{3} $ for any order $\alpha$ of partial differentiation by $x_{k} $ ; $k\in \{ 1,2,3\} $
\newline
First Partial Derivative:
\begin{align}
    \frac{\partial u_{i} }{\partial x_{k} } =\frac{\partial }{\partial x_{k} } \left(2\frac{d_{i} }{S^{2} } \right)
\label{eq1.58.1}\end{align}
\begin{align}
    \frac{\partial u_{i} }{\partial x_{k} } =\frac{2}{S^{4} } \left(\frac{\partial d_{i} }{\partial x_{k} } S^{2} -d_{i} 2S\frac{\partial S}{\partial x_{k} } \right)
\label{eq1.59}\end{align}
for $\vec{x}\in R^{3}$, $k=\{1,2,3\}$, $i=\{1,2,3\}$.
Per (\ref{eq1.1}):
\begin{align}
    \frac{\partial d_{i} }{\partial x_{k} } =\frac{\partial }{\partial x_{k} } x_{h(i-1)} -\frac{\partial }{\partial x_{k} } x_{h(i+1)} =\left\{\begin{array}{ccc} {1} & {\qquad ;k=h(i-1)} & {} \\ {-1} & {\qquad ;k=h(i+1)} & {} \\ {0} & {;k=i} & {} \end{array}\right.
\label{eq1.60}\end{align}
for $\vec{x}\in R^{3}$, $k=\{1,2,3\}$, $i=\{1,2,3\}$.
In the most general case, for the non-zero result when $k\ne i$, as per (\ref{eq1.60})
\begin{align}
    \frac{\partial d_{i} }{\partial x_{k} } =\pm 1
\label{eq1.60.1}\end{align}
once (\ref{eq1.60.1}) and (\ref{eq1.13}) are applied in (\ref{eq1.59})
\begin{align}
    \frac{\partial u_{i} }{\partial x_{k} } =\pm \frac{2}{S^{2} } -8\frac{d_{i} x_{k} }{S^{3} }
\label{eq1.60.2}\end{align}
for $\vec{x}\in R^{3}$, $k=\{1,2,3\}$, $i=\{1,2,3\}$.
Once (\ref{eq1.2}) is applied in (\ref{eq1.60.2})
\begin{align}
    \frac{\partial u_{i} }{\partial x_{k} } =\pm \frac{2}{\left(1+\left|\vec{x}\right|^{2} \right)^{2} } -8\frac{d_{i} x_{k} }{\left(1+\left|\vec{x}\right|^{2} \right)^{3} }
\label{eq1.61.1}\end{align}
let us determine the degree for $x$ of the first partial derivative:
\begin{align}
    \deg _{x} (\frac{\partial u_{i} }{\partial x_{k} } )=\deg (\pm \frac{2}{\left(1+\left|\vec{x}\right|^{2} \right)^{2} } -8\frac{d_{i} x_{k} }{\left(1+\left|\vec{x}\right|^{2} \right)^{3} } )
\label{eq1.62}\end{align}
for $\vec{x}\in R^{3}$, $k=\{1,2,3\}$, $i=\{1,2,3\}$.
\begin{align}
    \deg _{x} (\frac{\partial u_{i} }{\partial x_{k} } )=-4
\label{eq1.62.1}\end{align}
for $\vec{x}\in R^{3}$, $k=\{1,2,3\}$, $i=\{1,2,3\}$.
Note that both terms in the brackets of (\ref{eq1.62}) are of the same degree.
\newline
Second Partial Derivative:
Continuing partial differentiation from (\ref{eq1.60.2}):
\begin{align}
    \frac{\partial ^{2} u_{i} }{\partial x_{k}^{2} } =\frac{\partial }{\partial x_{k} } \left(\pm \frac{2}{S^{2} } -8\frac{d_{i} x_{k} }{S^{3} } \right)
\label{eq1.62.2}\end{align}
\begin{align}
    \frac{\partial ^{2} u_{i} }{\partial x_{k}^{2} } =\pm \frac{\partial }{\partial x_{k} } \left(\frac{2}{S^{2} } \right)-\frac{\partial }{\partial x_{k} } \left(8\frac{d_{i} x_{k} }{S^{3} } \right)
\label{eq1.62.3}\end{align}
\begin{align}
    \frac{\partial ^{2} u_{i} }{\partial x_{k}^{2} } =\mp \frac{4}{S^{3} } \frac{\partial S}{\partial x_{k} } -8\frac{\partial d_{i} }{\partial x_{k} } \frac{x_{k} }{S^{3} } -8d_{i} \frac{\partial }{\partial x_{k} } \left(\frac{x_{k} }{S^{3} } \right)
\label{eq1.62.4}\end{align}
for $\vec{x}\in R^{3}$, $k=\{1,2,3\}$, $i=\{1,2,3\}$.
In the most general case, for non-zero result when $k\ne i$, as per (\ref{eq1.60})
\begin{align}
    \frac{\partial d_{i} }{\partial x_{k} } =\pm 1
\label{eq1.62.5}\end{align}
applying (\ref{eq1.62.5}) in (\ref{eq1.62.1})
\begin{align}
    \frac{\partial ^{2} u_{i} }{\partial x_{k}^{2} } =\mp 8\frac{x_{k} }{S^{3} } \mp 8\frac{x_{k} }{S^{3} } -8d_{i} \frac{1}{S^{6} } \left(\frac{\partial x_{k} }{\partial x_{k} } S^{3} -x_{k} 3S^{2} \frac{\partial S}{\partial x_{k} } \right)
\label{eq1.63}\end{align}
for $\vec{x}\in R^{3}$, $k=\{1,2,3\}$, $i=\{1,2,3\}$.
Once (\ref{eq1.13}) and  (\ref{eq1.2}) are applied in (\ref{eq1.63})
\begin{align}
    \frac{\partial ^{2} u_{i} }{\partial x_{k}^{2} } =\frac{\mp 16x_{k} -8d_{i} }{\left(1+\left|\vec{x}\right|^{2} \right)^{3} } +48\frac{d_{i} x_{k}^{2} }{\left(1+\left|\vec{x}\right|^{2} \right)^{4} }
\label{eq1.64}\end{align}
for $\vec{x}\in R^{3}$, $k=\{1,2,3\}$, $i=\{1,2,3\}$.
The degree of $x$ for the second partial derivative, as per statement (\ref{eq1.64}) is:
\begin{align}
    \deg _{x} (\frac{\partial ^{2} u_{i} }{\partial x_{k}^{2} } )=\deg (\frac{\mp 16x_{k} -8d_{i} }{\left(1+\left|\vec{x}\right|^{2} \right)^{3} } +48\frac{d_{i} x_{k}^{2} }{\left(1+\left|\vec{x}\right|^{2} \right)^{4} } )
\label{eq1.65}\end{align}
\begin{align}
    \deg _{x} (\frac{\partial ^{2} u_{i} }{\partial x_{k}^{2} } )=-5
\label{eq1.65.1}\end{align}
Note that both terms in the brackets of statement (\ref{eq1.64}) are of the same degree for $x$.
In case that the first partial derivative is by $x_{k}$ and the second partial derivative is by $x_{j} $ for which $j\ne k$, starting from statement (\ref{eq1.60.2}):
\begin{align}
    \frac{\partial ^{2} u_{i} }{\partial x_{j} \partial x_{k} } =\frac{\partial }{\partial x_{j} } \left(\pm \frac{2}{S^{2} } -8\frac{d_{i} x_{k} }{S^{3} } \right)
\label{eq1.65.2}\end{align}
\begin{align}
    \frac{\partial ^{2} u_{i} }{\partial x_{j} \partial x_{k} } =\pm \frac{\partial }{\partial x_{j} } \left(\frac{2}{S^{2} } \right)-8\frac{\partial }{\partial x_{j} } \left(\frac{d_{i} x_{k} }{S^{3} } \right)
\label{eq1.65.3}\end{align}
\begin{align}
    \frac{\partial ^{2} u_{i} }{\partial x_{j} \partial x_{k} } =\mp \frac{4}{S^{3} } \frac{\partial S}{\partial x_{j} } -\frac{8}{S^{6} } \left(\frac{\partial (d_{i} x_{k} )}{\partial x_{j} } S^{3} -d_{i} x_{k} 3S^{2} \frac{\partial S}{\partial x_{j} } \right)
\label{eq1.66}\end{align}
for $\vec{x}\in R^{3}$, $k=\{1,2,3\}$, $i=\{1,2,3\}$, $j\ne k$.
As per (\ref{eq1.1}):
\begin{align}
    \frac{\partial (d_{i} x_{k} )}{\partial x_{j} } =\frac{\partial }{\partial x_{j} } \left((x_{h(i-1)} -x_{h(i+1)} )x_{k} \right)=\frac{\partial }{\partial x_{j} } \left(x_{h(i-1)} x_{k} -x_{h(i+1)} x_{k} \right)
\label{eq1.66.1}\end{align}
There are three scenarios to consider:
Scenario 1: $j=h(i-1)$:
\newline
\begin{align}
    \frac{\partial (d_{i} x_{k} )}{\partial x_{j} } =\frac{\partial }{\partial x_{j} } \left(x_{j} x_{k} -x_{h(i+1)} x_{k} \right)
\label{eq1.66.2}\end{align}
\begin{align}
    \frac{\partial (d_{i} x_{k} )}{\partial x_{j} } =x_{k}
\label{eq1.66.3}\end{align}
Scenario 2: $j=h(i+1)$:
\newline
\begin{align}
    \frac{\partial (d_{i} x_{k} )}{\partial x_{j} } =\frac{\partial }{\partial x_{j} } \left(x_{h(i-1)} x_{k} -x_{j} x_{k} \right)
\label{eq1.66.4}\end{align}
\begin{align}
    \frac{\partial (d_{i} x_{k} )}{\partial x_{j} } =-x_{k}
\label{eq1.66.5}\end{align}
Scenario 3: $j\ne h(i\pm 1)\wedge j\ne k$:
\newline
\begin{align}
    \frac{\partial (d_{i} x_{k} )}{\partial x_{j} } =\frac{\partial }{\partial x_{j} } \left(x_{h(i-1)} x_{k} -x_{h(i+1)} x_{k} \right)
\label{eq1.66.6}\end{align}
\begin{align}
    \frac{\partial (d_{i} x_{k} )}{\partial x_{j} } =0
\label{eq1.66.7}\end{align}
Based on the three scenarios (\ref{eq1.66.3}), (\ref{eq1.66.5}) and (\ref{eq1.66.7}), statement (\ref{eq1.66}) has the largest degree for $x$ when
\begin{align}
    j=h(i\pm 1)\vee j=k
\label{eq1.66.8}\end{align}
as otherwise, the result of differentiation by $x_{j} $ is zero.
As the scenario when $j=k$ is already covered with $\frac{\partial ^{2} u_{i} }{\partial x_{k}^{2} }$ with statement (\ref{eq1.65}), let us select scenario $j=h(i\pm 1)$ when as per (\ref{eq1.66.3}) and (\ref{eq1.66.5})
\begin{align}
    \frac{\partial (d_{i} x_{k} )}{\partial x_{j} } =\pm x_{k}
\label{eq1.67}\end{align}
once (\ref{eq1.67}) and (\ref{eq1.13}) are  applied in statement (\ref{eq1.66})
\begin{align}
    \frac{\partial ^{2} u_{i} }{\partial x_{j} \partial x_{k} } =\mp \frac{4}{S^{3} } 2x_{j} -\frac{8}{S^{6} } \left(\pm x_{k} S^{3} -3d_{i} x_{k} S^{2} 2x_{j} \right)
\label{eq1.67}\end{align}
\begin{align}
    \frac{\partial ^{2} u_{i} }{\partial x_{j} \partial x_{k} } =\mp \frac{16}{S^{3} } \left(x_{j} +x_{k} \right)+\frac{48}{S^{4} } d_{i} x_{k} x_{j}
\label{eq1.68}\end{align}
for $\vec{x}\in R^{3}$, $k=\{1,2,3\}$, $i=\{1,2,3\}$.
The degree for $x$ as per (\ref{eq1.68}) is
\begin{align}
    \deg _{x} (\frac{\partial ^{2} u_{i} }{\partial x_{j} \partial x_{k} } )=\deg _{x} (\mp \frac{16}{S^{3} } \left(x_{j} +x_{k} \right)+\frac{48}{S^{4} } d_{i} x_{k} x_{j} )
\label{eq1.68.1}\end{align}
As per (\ref{eq1.2}) and (\ref{eq1.68.1}):
\begin{align}
    \deg _{x} (\frac{\partial ^{2} u_{i} }{\partial x_{j} \partial x_{k} } )=-5
\label{eq1.68.2}\end{align}
for $\vec{x}\in R^{3}$, $k=\{1,2,3\}$, $i=\{1,2,3\}$.
Let us represent partial derivatives by $x$ in a more generic fashion:
\begin{align}
    \frac{\partial ^{\alpha } u_{i} }{\partial x^{\alpha } }
\label{eq1.68.3}\end{align}
for $\vec{x}\in R^{3}$, $k=\{1,2,3\}$, $i=\{1,2,3\}$.
Where $\alpha$ represents the order of partial derivative by $\alpha$-tuple of $x_{i} $  ;$i=\{ 1,2,3\}$. Also, instead of specific numerical values that appear with various terms, let us define $C \in R$ such that it represents the generalized form of any constant value, representing numerical values which do not contribute to this analysis.
Based on this, let us express the first partial derivative of velocity vector field components $u_{i}$;$i=\{ 1,2,3\}$, as per statement (\ref{eq1.61.1}) in a more general form:
\begin{align}
    \frac{\partial u_{i} }{\partial x} =\frac{C}{\left(1+\left|\vec{x}\right|^{2} \right)^{2} } +\frac{C }{\left(1+\left|\vec{x}\right|^{2} \right)^{3} } x^{2}
\label{eq1.69}\end{align}
The second partial derivative, as per (\ref{eq1.68.1}) can be represented as
\begin{align}
    \frac{\partial ^{2} u_{i} }{\partial x^{2} } =\frac{C }{\left(1+\left|\vec{x}\right|^{2} \right)^{3} } x+\frac{C }{\left(1+\left|\vec{x}\right|^{2} \right)^{4} } x^{3}
\label{eq1.70}\end{align}
The third partial derivative can be represented as:
\begin{align}
    \frac{\partial ^{3} u_{i} }{\partial x^{3} } =C\frac{\partial }{\partial x} \left(\frac{x}{S^{3} } \right)+C\frac{\partial }{\partial x} \left(\frac{x^{3} }{S^{4} } \right)
\label{eq1.71}\end{align}
\begin{align}
    \frac{\partial ^{3} u_{i} }{\partial x^{3} } =\frac{C}{S^{6} } \left(\frac{\partial x}{\partial x} S^{3} -x3S^{2} \frac{\partial S}{\partial x} \right)+\frac{C}{S^{8} } \left(3x^{2} S^{4} -x^{3} 4S^{3} \frac{\partial S}{\partial x} \right)
\label{eq1.72}\end{align}
once (\ref{eq1.13})is applied in (\ref{eq1.72})
\begin{align}
    \frac{\partial ^{3} u_{i} }{\partial x^{3} } =\frac{C}{S^{6} } \left(S^{3} -3xS^{2} 2x\right)+\frac{C}{S^{8} } \left(3x^{2} S^{4} -4x^{3} S^{3} 2x\right)
\label{eq1.74}\end{align}
\begin{align}
    \frac{\partial ^{3} u_{i} }{\partial x^{3} } =\frac{C}{S^{3} } +\frac{C}{S^{4} } x^{2} +\frac{C}{S^{5} } x^{4}
\label{eq1.77}\end{align}
for $\vec{x}\in R^{3}$, $k=\{1,2,3\}$, $i=\{1,2,3\}$.
The degree of $x$ as per (\ref{eq1.77}) is:
\begin{align}
    \deg _{x} (\frac{\partial ^{3} u_{i} }{\partial x^{3} } )=\deg _{x} (\frac{C}{S^{3} } +\frac{C}{S^{4} } x^{2} +\frac{C}{S^{5} } x^{4} )
\label{eq1.78}\end{align}
\begin{align}
    \deg _{x} (\frac{\partial ^{3} u_{i} }{\partial x^{3} } )=-6
\label{eq1.79}\end{align}
Notably, all terms in statement (\ref{eq1.78}) are of the same degree.
Let us perform one additional partial derivative. The fourth partial derivative can be represented as
\begin{align}
    \frac{\partial }{\partial x} \left(\frac{\partial ^{3} u_{i} }{\partial x^{3} } \right)=\frac{\partial }{\partial x} \left(\frac{C}{S^{3} } +\frac{C}{S^{4} } x^{2} +\frac{C}{S^{5} } x^{4} \right)
\label{eq1.80}\end{align}
\begin{align}
    \frac{\partial ^{4} u_{i} }{\partial x^{4} } =\frac{\partial }{\partial x} \left(\frac{C}{S^{3} } +\frac{C}{S^{4} } x^{2} +\frac{C}{S^{5} } x^{4} \right)
\label{eq1.81}\end{align}
The partial derivatives of each individual term of statement (\ref{eq1.81}) in the brackets are:
\begin{align}
    \frac{\partial }{\partial x} \left(\frac{C}{S^{3} } \right)=-3\frac{C}{S^{6} } S^{2} \frac{\partial S}{\partial x} =\frac{C}{S^{4} } \frac{\partial S}{\partial x} ==\frac{C}{S^{4} } 2x=\frac{C}{S^{4} } x
\label{eq1.82}\end{align}
\begin{align}
    \frac{\partial }{\partial x} \left(\frac{C}{S^{4} } x^{2} \right)=\frac{C}{S^{8} } 2xS^{4} -\frac{C}{S^{8} } x^{2} 4S^{3} \frac{\partial S}{\partial x} =\frac{C}{S^{4} } x-\frac{C}{S^{5} } 4x^{2} 2x=\frac{C}{S^{4} } x+\frac{C}{S^{5} } x^{3}
\label{eq1.83}\end{align}
\begin{align}
    \frac{\partial }{\partial x} \left(\frac{C}{S^{5} } x^{4} \right)=\frac{C}{S^{10} } 4x^{3} S^{5} -\frac{C}{S^{10} } x^{4} 5S^{4} \frac{\partial S}{\partial x} =\frac{C}{S^{5} } x^{3} -\frac{C}{S^{6} } x^{4} 2x=\frac{C}{S^{5} } x^{3} +\frac{C}{S^{6} } x^{5}
\label{eq1.84}\end{align}
Once all three partial derivatives as per statements (\ref{eq1.82}), (\ref{eq1.83}), (\ref{eq1.84}) are summed up together:
\begin{align}
    \frac{\partial ^{4} u_{i} }{\partial x^{4} } =\frac{C}{S^{4} } x+\frac{C}{S^{5} } x^{3} +\frac{C}{S^{6} } x^{5}
\label{eq1.85}\end{align}
for any $\vec{x}\in R^{3}$.
The degree of $x$ as per (\ref{eq1.85}) is:
\begin{align}
    \deg _{x} (\frac{\partial ^{4} u_{i} }{\partial x^{4} } )=\deg _{x} (\frac{C}{S^{4} } x+\frac{C}{S^{5} } x^{3} +\frac{C}{S^{6} } x^{5} )=-7
\label{eq1.86}\end{align}
Also, notably, all terms above are of the same degree.
\newline
Based on statements (\ref{eq1.62.1}), (\ref{eq1.65.1}), (\ref{eq1.79}), (\ref{eq1.86}), we can conclude that the $\alpha $-th order of the partial differential of the fluid velocity vector field component $u_{i}$ by $x$ can be expressed as:
\begin{align}
    \deg _{x} (\frac{\partial ^{\alpha } u_{i} }{\partial x^{\alpha } } )=-\alpha -3
\label{eq1.86.1}\end{align}
for any $\{\alpha > 0\}\in N$, $\vec{x}\in R^{3}$, $i=\{1,2,3\}$.
Based on statements (\ref{eq1.62.1}), (\ref{eq1.65.1}), (\ref{eq1.79}), (\ref{eq1.86}), the partial differentiation by $x$ of any order $\alpha$ can be represented as:
\begin{align}
    \frac{\partial ^{\alpha } u_{i} }{\partial x^{\alpha } } =\frac{C }{\left(1+\left|\vec{x}\right|^{2} \right)^{\left\lfloor \frac{\left|\alpha \right|}{2} \right\rfloor +2} } +\sum _{j=1}^{\left\lfloor \frac{\left|\alpha \right|-1}{2} \right\rfloor +2}C\frac{ x^{k_{j} } }{\left(1+\left|\vec{x}\right|^{2} \right)^{m_{j} } }
\label{eq1.94}\end{align}
for any $\{\alpha > 0\}\in N$, $x \in R^{3}; i=\{ 1,2,3\}$ and where $x^{k_{j} }$ represents $k_{j} $ - tuple of any of $x_{i} $ ; $j=\{ 1,2,3\}$
\newline
Also, notably, that each of the terms in statement (\ref{eq1.94}) are of the same degree:
\begin{align}
    \deg _{x} (C\frac{ x^{k_{j} } }{\left(1+\left|\vec{x}\right|^{2} \right)^{m_{j} } } )=-\alpha -3
\label{eq1.95}\end{align}
Then, for each of the terms of sum in statement (\ref{eq1.94}), limit for $|\vec{x}|\to \infty$:
\begin{align}
    \lim _{\left|\vec{x}\right|\to \infty } C \frac{ x^{k_{j} } }{\left(1+\left|\vec{x}\right|^{2} \right)^{m_{j} } }=0
\label{eq1.96}\end{align}
At the coordinate origin $\vec{x}=\vec{0}$, each of the terms of the sum in statement (\ref{eq1.94}) is zero:
\begin{align}
    \left[C \frac{x^{k_{j} } }{\left(1+\left|\vec{x}\right|^{2} \right)^{m_{j} } } \right]_{\vec{x}=\vec{0}} =0
\label{eq1.97}\end{align}
Let us name $C_{\alpha } (i,j)\in R$ such that
\begin{align}
    C_{\alpha } (i,j)=C \frac{ x^{k_{j} } }{\left(1+\left|\vec{x}\right|^{2} \right)^{m_{j} } }
\label{eq1.100}\end{align}
In order to verify that $C_{\alpha } (i,j)$ has to be finite, let us assume the opposite:
\begin{align}
    C_{\alpha } (i,j)=\infty
\label{eq1.100.1}\end{align}
As per statement (\ref{eq1.96}) we concluded that for $\left|\vec{x}\right| \to \infty$, the value of each term of statement (\ref{eq1.94}) converges to zero, and at the coordinate origin, as per (\ref{eq1.96}), each of the terms has value zero. Then, for any selected position $\vec{x}\in R^{3} $, each of the terms can be evaluated as infinite only in case that at least one of terms' denominator is equal to zero:
\begin{align}
    1+\left|\vec{x}\right|^{2} =0
\label{eq1.102}\end{align}
However, $1+\left|\vec{x}\right|^{2} $ is definite positive and cannot be zero. Its minimal value is 1. Based on this, we can conclude that statement (\ref{eq1.100.1}) is impossible, and therefore $C_{\alpha } (i,j)$ must be finite.
\newline
As $1+\left|\vec{x}\right|^{2} $ is definite positive than the first term of statement (\ref{eq1.94})
\begin{align}
    \frac{a_{0} }{\left(1+\left|\vec{x}\right|^{2} \right)^{\left\lfloor \frac{\left|\alpha \right|}{2} \right\rfloor +2} }
\label{eq1.104}\end{align}
also has to be finite for any $\{\alpha>0\} \in N$.
\newline
Based on these two conclusions, the sum of all terms of statement (\ref{eq1.94}) must be finite as well:
\begin{align}
    \frac{\partial ^{\alpha } u_{i} }{\partial x^{\alpha } } \ne \infty
\label{eq1.105}\end{align}
Based on (\ref{eq1.105}), there must be some finite $\{ C_{\alpha } (i)\ge 0\} \in R$ such that
\begin{align}
    \left|\frac{\partial ^{\alpha } u_{i} }{\partial x^{\alpha } } \right|=\left|\frac{a_{0} }{\left(1+\left|\vec{x}\right|^{2} \right)^{\left\lfloor \frac{\left|\alpha \right|}{2} \right\rfloor +2} } +\sum _{j=1}^{n(\alpha )}\frac{a_{j} x^{k_{j} } }{\left(1+\left|\vec{x}\right|^{2} \right)^{m_{j} } }  \right|\le C_{\alpha } (i) ; i=\{ 1,2,3\}
\label{eq1.106}\end{align}
then for
\begin{align}
    \left|\partial _{x}^{\alpha } \vec{u}(\vec{x})\right|=\sqrt{\sum _{i=1}^{3}\left|\frac{\partial ^{\alpha } u_{i} (\vec{x})}{\partial x^{\alpha } } \right|^{2}  }
\label{eq1.107}\end{align}
as per (\ref{eq1.106}) and (\ref{eq1.107})
\begin{align}
    \sum _{i=1}^{3}\left|\frac{\partial ^{\alpha } u_{i} (\vec{x})}{\partial x^{\alpha } } \right|^{2}  \le \sum _{i=1}^{3}C_{\alpha } (i)
\label{eq1.108}\end{align}
based on this, there must exists $C_{\alpha } \in R$ such that
\begin{align}
    C_{\alpha } \ge \sqrt{\sum _{i=1}^{3}C_{\alpha } (i) }
\label{eq1.109}\end{align}
based on (\ref{eq1.109}, (\ref{eq1.108}) and (\ref{eq1.107}):
\begin{align}
    \left|\partial _{x}^{\alpha } \vec{u}(\vec{x})\right|\le C_{\alpha }
\label{eq1.110}\end{align}
for any $\{\alpha > 0\}\in N$, $x \in R^{3}$.
\newline\newline
\textbf{Continuous differentiability of force field $\vec{f}(\vec{x},t)$}
\newline
As per the theorem statement, $f_{3}$, the component of $\vec{f}$ is
\begin{align}
    f_{3} (\vec{x},t)=\frac{1}{1+t^{2} \left(\sum _{j=1}^{3}x_{j}  \right)^{2} }
\label{eq1.120}\end{align}
for any $x \in R^{3}$, $t \ge 0$.
Let us define the denominator of statement (\ref{eq1.120}) as:
\begin{align}
    B=1+t^{2} \left(\sum _{j=1}^{3}x_{j}  \right)^{2}
\label{eq1.121}\end{align}
applying (\ref{eq1.121}) in (\ref{eq1.120}):
\begin{align}
    f_{3} (\vec{x},t)=\frac{1}{B}
\label{eq1.122}\end{align}
Let us apply partial differential by $t$ on (\ref{eq1.122}):
\begin{align}
    \frac{\partial }{\partial t} f_{3} (\vec{x},t)=\frac{\partial }{\partial t} \frac{1}{B}
\label{eq1.123}\end{align}
\begin{align}
    \frac{\partial }{\partial t} f_{3} (\vec{x},t)=-\frac{1}{B^{2} } \frac{\partial B}{\partial t}
\label{eq1.124}\end{align}
as per (\ref{eq1.121}):
\begin{align}
    \frac{\partial B}{\partial t} =\frac{\partial }{\partial t} \left(1+t^{2} \left(\sum _{j=1}^{3}x_{j}  \right)^{2} \right)=2t\left(\sum _{j=1}^{3}x_{j}  \right)^{2}
\label{eq1.125}\end{align}
applying (\ref{eq1.125})in (\ref{eq1.124}):
\begin{align}
    \frac{\partial }{\partial t} f_{3} (\vec{x},t)=-\frac{2t}{B^{2} } \left(\sum _{j=1}^{3}x_{j}  \right)^{2}
\label{eq1.126}\end{align}
Let us define
\begin{align}
    H=\sum _{j=1}^{3}x_{j}
\label{eq1.127}\end{align}
applying (\ref{eq1.127})in (\ref{eq1.126}):
\begin{align}
    \frac{\partial }{\partial t} f_{3} (\vec{x},t)=-\frac{2t}{B^{2} } H^{2}
\label{eq1.128}\end{align}
let us find the second partial derivative by $t$:
\begin{align}
    \frac{\partial ^{2} }{\partial t^{2} } f_{3} (\vec{x},t)=\frac{\partial }{\partial t} \left(-\frac{2t}{B^{2} } SH^{2} \right)
\label{eq1.129}\end{align}
as $H=\sum _{j=1}^{3}x_{j}$ is not in the function of time:
\begin{align}
    \frac{\partial ^{2} }{\partial t^{2} } f_{3} (\vec{x},t)=-2H^{2} \frac{\partial }{\partial t} \left(\frac{t}{B^{2} } \right)
\label{eq1.130}\end{align}
\begin{align}
    \frac{\partial ^{2} }{\partial t^{2} } f_{3} (\vec{x},t)=-2H^{2} \frac{1}{B^{4} } \left(B^{2} \frac{\partial t}{\partial t} -t2B\frac{\partial B}{\partial t} \right)
\label{eq1.132}\end{align}
as per statements (\ref{eq1.121}) and (\ref{eq1.127}), $\frac{\partial B}{\partial t}$ is:
\begin{align}
    \frac{\partial B}{\partial t} =2t\left(\sum _{j=1}^{3}x_{j}  \right)^{2} =2tH^{2}
\label{eq1.133}\end{align}
once (\ref{eq1.133}) is applied in (\ref{eq1.132}):
\begin{align}
    \frac{\partial ^{2} }{\partial t^{2} } f_{3} (\vec{x},t)=-2H^{2} \frac{1}{B^{4} } \left(B^{2} -t2B2tH^{2} \right)
\label{eq1.134}\end{align}
\begin{align}
    \frac{\partial ^{2} }{\partial t^{2} } f_{3} (\vec{x},t)=-2H^{2} \frac{1}{B^{4} } \left(B^{2} -4t^{2} BH^{2} \right)
\label{eq1.135}\end{align}
\begin{align}
    \frac{\partial ^{2} }{\partial t^{2} } f_{3} (\vec{x},t)=-2\frac{H^{2} }{B^{2} } +8\frac{t^{2} H^{4} }{B^{3} }
\label{eq1.136}\end{align}
Similarly, the third partial derivative by $t$ is:
\begin{align}
    \frac{\partial ^{3} }{\partial t^{3} } f_{3} (\vec{x},t)=-48\frac{t^{3} H^{6} }{B^{4} } +24\frac{tH^{4} }{B^{3} }
\label{eq1.137}\end{align}
The general form of $m$-th  partial derivative by $t$ can be expressed as:
\begin{align}
    \frac{\partial ^{m} }{\partial t^{m} } f_{3} (\vec{x},t)=\sum _{l=1}^{\left\lfloor \frac{m}{2} \right\rfloor +1}(-1)^{\left\| m-\left\lfloor \frac{m}{2} +1\right\rfloor +l\right\| _{2} } \frac{t^{\left\| m\right\| _{2} +2(l-1)} H^{2\left(m-\left\lfloor \frac{m}{2} +1\right\rfloor +l\right)} }{B^{\left\lfloor \frac{m-1}{2} \right\rfloor +l+1} }
\label{eq1.138}\end{align}
for any $x \in R^{3}$, $t \ge 0$,
where
\begin{align}
    h=\left\| m\right\| _{2} +2(l-1)
\label{eq1.139}\end{align}
\begin{align}
    e=2\left(m-\left\lfloor \frac{m}{2} +1\right\rfloor +l\right)
\label{eq1.140}\end{align}
\begin{align}
    s=\left\lfloor \frac{m-1}{2} \right\rfloor +l+1
\label{eq1.141}\end{align}
once (\ref{eq1.139}), (\ref{eq1.140}), (\ref{eq1.141})are applied in (\ref{eq1.138}):
\begin{align}
    \frac{\partial ^{m} }{\partial t^{m} } f_{3} (\vec{x},t)=\sum _{l=1}^{\left\lfloor \frac{m}{2} \right\rfloor +1}(-1)^{\left\| \frac{e}{2} \right\| _{2} } \frac{t^{h} S^{e} }{B^{s} }
\label{eq1.142}\end{align}
for any $x \in R^{3}$, $t \ge 0$, $\{m>0\} \in N$.
In order to demonstrate that statement (\ref{eq1.142}) correctly represents the partial derivation for any order of differentiation $\{m>0\} \in N$, let us apply a few orders of partial derivatives by $t$ using statement above:
\begin{align}
    \frac{\partial }{\partial t} f_{3} (\vec{x},t)=-2\frac{tH^{2} }{B^{2} }
\label{eq1.143}\end{align}
\begin{align}
    \frac{\partial ^{2} }{\partial t^{2} } f_{3} (\vec{x},t)=-2\frac{H^{2} }{B^{2} } +8\frac{t^{2} H^{4} }{B^{3} }
\label{eq1.144}\end{align}
\begin{align}
    \frac{\partial ^{3} }{\partial t^{3} } f_{3} (\vec{x},t)=-48\frac{t^{3} H^{6} }{B^{4} } +24\frac{tH^{4} }{B^{3} }
\label{eq1.145}\end{align}
\begin{align}
    \frac{\partial ^{4} }{\partial t^{4} } f_{3} (\vec{x},t)=384\frac{t^{4} H^{8} }{B^{5} } -288\frac{t^{2} H^{6} }{B^{4} } +24\frac{H^{4} }{B^{3} }
\label{eq1.146}\end{align}
\begin{align}
    \frac{\partial ^{5} }{\partial t^{5} } f_{3} (\vec{x},t)=-3840\frac{t^{5} H^{10} }{B^{6} } +3840\frac{t^{3} H^{8} }{B^{5} } -720\frac{tH^{6} }{B^{4} }
\label{eq1.147}\end{align}
Statements from (\ref{eq1.143}) to (\ref{eq1.147}) match the partial derivatives by $t$ performed in incremental fashion on the basis of the previous order of the partial derivative, confirming that (\ref{eq1.142}) represents the general form of partial derivative by $t$ or any order $\{m>0\} \in N$.
\newline
As in statement (\ref{eq1.142}) $B$ as per (\ref{eq1.121}) is definite positive, then $\frac{\partial ^{m} }{\partial t^{m} } f_{3} (\vec{x},t)$ is continuous for any $\{ m > 0\} \in N$ .As $\vec{f}(\vec{x},t)=(0,0,f_{3} )$ we can conclude that  $\frac{\partial ^{m} }{\partial t^{m} } \vec{f}(\vec{x},t)$ is also continuous for any $\{ m > 0\} \in N$ . Therefore, we can conclude that the vector field $\vec{f}(\vec{x},t)$ is continuously differentiable  $\vec{f}(\vec{x},t)\in C^{\infty } $ for any position $\vec{x}\in R^{3} $ and $t\ge 0$.
\newline\newline
\textbf{Determining $\deg_{x} (\frac{\partial ^{m} }{\partial t^{m}}\vec{f})$ and $\deg_{t} (\frac{\partial ^{m} }{\partial t^{m}}\vec{f})$}
\newline
Let us start with an analysis of $\deg _{x}(\frac{\partial ^{m}}{\partial t^{m}}\vec{f})$. Per (\ref{eq1.142}), expanding $H$ and $B$ as per (\ref{eq1.127}) and (\ref{eq1.121}):
\begin{align}
    \deg _{x} (\frac{\partial ^{m}f_{3}}{\partial t^{m} } )=Max(\deg _{x} (\frac{t^{k} \left(\sum _{j=1}^{3}x_{j}  \right)^{e} }{\left(1+t^{2} \left(\sum _{j=1}^{3}x_{j}  \right)^{2} \right)^{s} } ))_{l=\{ 1,\left\lfloor \frac{m}{2} \right\rfloor +1\} }
\label{eq1.150}\end{align}
\begin{align}
    \deg _{x} (\frac{\partial ^{m}f_{3}}{\partial t^{m} })=Max(\deg _{x} (\frac{\left(\sum _{j=1}^{3}x_{j}  \right)^{e} }{\left(\sum _{j=1}^{3}x_{j}  \right)^{2s} } ))_{l=\{ 1,\left\lfloor \frac{m}{2} \right\rfloor +1\} }
\label{eq1.150.2}\end{align}
\begin{align}
    \deg _{x} (\frac{\partial ^{m}f_{3} }{\partial t^{m}})=Max\left(e-2s\right)_{l=\{ 1,\left\lfloor \frac{m}{2} \right\rfloor +1\} }
\label{eq1.150.3}\end{align}
expanding $e$ and $s$, as per (\ref{eq1.140}) and (\ref{eq1.141}) in (\ref{eq1.150.3}):
\begin{align}
    \deg _{x} (\frac{\partial ^{m}f_{3} }{\partial t^{m}})=Max\left(2\left(m-\left\lfloor \frac{m}{2} +1\right\rfloor +l\right)-2(\left\lfloor \frac{m-1}{2} \right\rfloor +l+1)\right)_{l=\{ 1,\left\lfloor \frac{m}{2} \right\rfloor +1\} }
\label{eq1.152}\end{align}
\begin{align}
    \deg _{x} (\frac{\partial ^{m}f_{3}}{\partial t^{m}})=Max\left(2m-2\left\lfloor \frac{m}{2} +1\right\rfloor -2\left\lfloor \frac{m-1}{2} \right\rfloor -2\right)_{l=\{ 1,\left\lfloor \frac{m}{2} \right\rfloor +1\} }
\label{eq1.154}\end{align}
Statement (\ref{eq1.154}) is not in function of $l$, then (\ref{eq1.154}) is equivalent to
\begin{align}
    \deg _{x} (\frac{\partial ^{m}f_{3}}{\partial t^{m} })=2m-2\left\lfloor \frac{m}{2} +1\right\rfloor -2\left\lfloor \frac{m-1}{2} \right\rfloor -2
\label{eq1.155}\end{align}
Let us evaluate statement (\ref{eq1.155}) for a few consecutive values of $m$:
\newline
\[m=1 \Longrightarrow \deg _{x}(\frac{\partial f_{3}}{\partial t })=2-2-2=-2\]
\[m=2 \Longrightarrow \deg _{x}(\frac{\partial ^{2}f_{3}}{\partial t^{2}})=4-4-2=-2\]
\[m=3 \Longrightarrow \deg _{x}(\frac{\partial ^{3}f_{3}}{\partial t^{3}})=6-4-2-2=-2\]
\[m=4 \Longrightarrow \deg _{x}(\frac{\partial ^{4}f_{3}}{\partial t^{4}})=8-6-2-2=-2\]
\[m=5 \Longrightarrow \deg _{x}(\frac{\partial ^{5}f_{3}}{\partial t^{5}})=10-6-4-2=-2\]
Therefore, for any order $m$ of partial differentiations by time, the degree by $x$ for $\frac{\partial ^{m} }{\partial t^{m} } f_{3} (\vec{x},t)$can be expressed as
\begin{align}
    \deg _{x} (\frac{\partial ^{m} }{\partial t^{m} } f_{3} (\vec{x},t))=-2
\label{eq1.156}\end{align}
for any $\{m>0\} \in N$.
\newline
As already concluded, the degree for $x$ is not in function of $l$ as used in sum of statement (\ref{eq1.142}) meaning that each of the terms $\frac{t^{h} H^{e} }{B^{s} }$ of the sum are of the same degree, or in other words:
\begin{align}
    \deg _{x} (\frac{\partial ^{m} }{\partial t^{m} } f_{3} (\vec{x},t))=\deg _{x}( \frac{t^{h} H^{e} }{B^{s}})=-2
\label{eq1.156.1}\end{align}
\newline
As the degree for $x$ is $-2$ regardless of $m$, then for $\vec{x}$ approaching infinity $\left|\vec{x}\right|\to \infty $
\begin{align}
    \lim _{\left|\vec{x}\right|\to \infty } \frac{\partial ^{m} }{\partial t^{m} } f_{3} (\vec{x},t)=0
\label{eq1.157}\end{align}
limit of $m$-th order of partial derivative by time must be equal to zero.
\noindent Also, at the coordinate origin $\vec{x}=\vec{0}$ :
\begin{align}
    \frac{\partial ^{m} }{\partial t^{m} } f_{3} (\vec{x},t)|_{\vec{x}=\vec{0}} =\sum _{l=1}^{\left\lfloor \frac{m}{2} \right\rfloor +1}(-1)^{\left\| \frac{e}{2} \right\| _{2} } \frac{t^{h} 0^{e} }{B^{s} }  =0
\label{eq1.158}\end{align}
$m$-th order of partial derivative by time is zero as well.
\newline
Now, let us check the degree of $t$ as per statement (\ref{eq1.150}):
\begin{align}
    \deg _{t} (\frac{\partial ^{m}f_{3}}{\partial t^{m}})=Max(\deg _{t} (\frac{t^{h} \left(\sum _{j=1}^{3}x_{j}  \right)^{e} }{\left(1+t^{2} \left(\sum _{j=1}^{3}x_{j}  \right)^{2} \right)^{s} } ))_{l=\{ 1,\left\lfloor \frac{m}{2} \right\rfloor +1\} }
\label{eq1.159}\end{align}
\begin{align}
    \deg _{t} (\frac{\partial ^{m}f_{3}}{\partial t^{m}})=Max(\deg _{t} (\frac{t^{h} }{t^{2s} } ))_{l=\{ 1,\left\lfloor \frac{m}{2} \right\rfloor +1\} }
\label{eq1.160}\end{align}
expanding $h$ and $s$, as per (\ref{eq1.139}) and (\ref{eq1.141}) in (\ref{eq1.160}):
\begin{align}
    \deg _{t} (\frac{\partial ^{m}f_{3}}{\partial t^{m}})=Max(\deg _{t} (\frac{t^{\left\| m\right\| _{2} +2(l-1)} }{t^{2(\left\lfloor \frac{m-1}{2} \right\rfloor +l+1)} } ))_{l=\{ 1,\left\lfloor \frac{m}{2} \right\rfloor +1\} }
\label{eq1.161}\end{align}
\begin{align}
    \deg _{t} (\frac{\partial ^{m}f_{3}}{\partial t^{m}})=Max(\left\| m\right\| _{2} +2(l-1)-2(\left\lfloor \frac{m-1}{2} \right\rfloor +l+1))_{l=\{ 1,\left\lfloor \frac{m}{2} \right\rfloor +1\} }
\label{eq1.162}\end{align}
\begin{align}
    \deg _{t} (\frac{\partial ^{m}f_{3}}{\partial t^{m}})=Max(\left\| m\right\| _{2} +2l-2-2\left\lfloor \frac{m-1}{2} \right\rfloor -2l-2)_{l=\{ 1,\left\lfloor \frac{m}{2} \right\rfloor +1\} }
\label{eq1.163}\end{align}
\begin{align}
    \deg _{t} (\frac{\partial ^{m}f_{3}}{\partial t^{m}})=Max(\left\| m\right\| _{2} -2\left\lfloor \frac{m-1}{2} \right\rfloor -4)_{l=\{ 1,\left\lfloor \frac{m}{2} \right\rfloor +1\} }
\label{eq1.164}\end{align}
As per (\ref{eq1.164}), the degree of $t$ does not depend on $l$. Therefore, (\ref{eq1.164}) is equivalently represented as:
\begin{align}
    \deg _{t} (\frac{\partial ^{m}f_{3}}{\partial t^{m}} )=\left\| m\right\| _{2} -2\left\lfloor \frac{m-1}{2} \right\rfloor -4
\label{eq1.165}\end{align}
for any $\{m>0\} \in N$.
Let us evaluate statement (\ref{eq1.165}) for a few consecutive values of $m$:
\[\deg _{t} (\frac{\partial ^{m}f_{3}}{\partial t^{m}})_{m=1} =\left[\left\| m\right\| _{2} -2\left\lfloor \frac{m-1}{2} \right\rfloor -4\right]_{m=1} =1-0-4=-3\]
\[\deg _{t} (\frac{\partial ^{m}f_{3}}{\partial t^{m}})_{m=2} =\left[\left\| m\right\| _{2} -2\left\lfloor \frac{m-1}{2} \right\rfloor -4\right]_{m=2} =0-0-4=-4\]
\[\deg _{t} (\frac{\partial ^{m}f_{3}}{\partial t^{m}})_{m=3} =\left[\left\| m\right\| _{2} -2\left\lfloor \frac{m-1}{2} \right\rfloor -4\right]_{m=3} =1-2-4=-5\]
\[\deg _{t} (\frac{\partial ^{m}f_{3}}{\partial t^{m}})_{m=4} =\left[\left\| m\right\| _{2} -2\left\lfloor \frac{m-1}{2} \right\rfloor -4\right]_{m=4} =0-2-4=-6\]
\[\deg _{t} (\frac{\partial ^{m}f_{3}}{\partial t^{m}})_{m=5} =\left[\left\| m\right\| _{2} -2\left\lfloor \frac{m-1}{2} \right\rfloor -4\right]_{m=5} =1-4-4=-7\]
Based on this, the general form of  $\deg _{x}(\frac{\partial ^{m}}{\partial t^{m}}\vec{f})$ can be expressed as:
\begin{align}
    \deg _{t} (\frac{\partial ^{m}f_{3} }{\partial t^{m} })=-m-2
\label{eq1.166}\end{align}
for any $\{m>0\} \in N$.
As per statements (\ref{eq1.142}) and (\ref{eq1.166}) for $t$ approaching infinity $t\to \infty $
\begin{align}
    \lim _{t\to \infty } \frac{\partial ^{m} }{\partial t^{m} } f_{3} (\vec{x},t)=0
\label{eq1.167}\end{align}
for any $\{m>0\} \in N$.
The limit of the $m$-th order of the partial derivative by $t$ must be equal to zero.
Also, when $t=0$ :
\begin{align}
    \frac{\partial ^{m} }{\partial t^{m} } f_{3} (\vec{x},t)|_{t=0} =0
\label{eq1.168}\end{align}
for any $\{m>0\} \in N$.
\newline\newline
\textbf{Determining $ \deg _{x} (\frac{\partial ^{\alpha } }{\partial x^{\alpha } } \frac{\partial ^{m} }{\partial t^{m} }\vec{f})$}
\newline
Now, beginning with statement (\ref{eq1.142}), let us perform the partial differentiation by $x$ :
\begin{align}
    \frac{\partial }{\partial x} \frac{\partial ^{m} }{\partial t^{m} } f_{3} (\vec{x},t)=\frac{\partial }{\partial x} \sum _{l=1}^{\left\lfloor \frac{m}{2} \right\rfloor +1}(-1)^{\left\| \frac{e}{2} \right\| _{2} } \frac{t^{h} H^{e} }{B^{s} }
\label{eq1.169}\end{align}
\begin{align}
    \frac{\partial }{\partial x} \frac{\partial ^{m} }{\partial t^{m} } f_{3} (\vec{x},t)=\sum _{l=1}^{\left\lfloor \frac{m}{2} \right\rfloor +1}(-1)^{\left\| \frac{e}{2} \right\| _{2} } t^{h} \frac{\partial }{\partial x} \frac{H^{e} }{B^{s} }
\label{eq1.169.1}\end{align}
let us expand $\frac{\partial }{\partial x} \frac{H^{e} }{B^{s} }$ in statement (\ref{eq1.169.1})
\begin{align}
    \frac{\partial }{\partial x} \frac{H^{e} }{B^{s} } =\frac{1}{B^{2s} } \left(B^{s} \frac{\partial }{\partial x} H^{e} -H^{e} \frac{\partial }{\partial x} B^{s} \right)
\label{eq1.170}\end{align}
\begin{align}
    \frac{\partial }{\partial x} \frac{H^{e} }{B^{s} } =\frac{1}{B^{2s} } \left(B^{s} eH^{e-1} \frac{\partial H}{\partial x} -H^{e} sB^{s-1} \frac{\partial B}{\partial x} \right)
\label{eq1.171}\end{align}
applying partial derivative $ \frac{\partial}{\partial x}$ on $H$, as defined per (\ref{eq1.127}):
\begin{align}
    \frac{\partial H}{\partial x} =\frac{\partial }{\partial x} \sum _{j=1}^{3}x_{j}  =1
\label{eq1.172}\end{align}
applying partial derivative $ \frac{\partial}{\partial x}$ on $B$ as defined per (\ref{eq1.121}):
\begin{align}
    \frac{\partial B}{\partial x} =\frac{\partial}{\partial x}(1+t^{2} \left(\sum _{j=1}^{3}x_{j}\right)^{2})=2t^{2} \sum _{j=1}^{3}x_{j}  =2t^{2} H
\label{eq1.173}\end{align}
then applying (\ref{eq1.173}) and (\ref{eq1.172}) on (\ref{eq1.171}):
\begin{align}
    \frac{\partial }{\partial x} \frac{H^{e} }{B^{s} } =\frac{1}{B^{2s} } \left(B^{s} eH^{e-1} -H^{e} sB^{s-1} 2t^{2} H\right)
\label{eq1.174}\end{align}
\begin{align}
    \frac{\partial }{\partial x} \frac{H^{e} }{B^{s} } =\frac{1}{B^{2s} } \left(eB^{s} H^{e-1} -2st^{2} B^{s-1} H^{e+1} \right)
\label{eq1.175}\end{align}
\begin{align}
    \frac{\partial }{\partial x} \frac{H^{e} }{B^{s} } =\frac{eH^{e-1} }{B^{s} } -\frac{2st^{2} H^{e+1} }{B^{s+1} }
\label{eq1.177}\end{align}
once (\ref{eq1.177}) is applied in (\ref{eq1.169.1}):
\begin{align}
    \frac{\partial }{\partial x} \frac{\partial ^{m} }{\partial t^{m} } f_{3} (\vec{x},t)=\sum _{l=1}^{\left\lfloor \frac{m}{2} \right\rfloor +1}(-1)^{\left\| \frac{e}{2} \right\| _{2} } \left(e\frac{t^{h} H^{e-1} }{B^{s} } -2s\frac{t^{h+2} H^{e+1} }{B^{s+1} } \right)
\label{eq1.179}\end{align}
Let us determine the degree for $x$ as per statement (\ref{eq1.179}).
\begin{align}
    \deg _{x} (\frac{\partial }{\partial x} \frac{\partial ^{m} }{\partial t^{m} } f_{3} (\vec{x},t))=\deg _{x} (\sum _{l=1}^{\left\lfloor \frac{m}{2} \right\rfloor +1}\left(e\frac{t^{h} H^{e-1} }{B^{s} } -2s\frac{t^{h+2} H^{e+1} }{B^{s+1} } \right) )
\label{eq1.180}\end{align}
Let us determine the degree of the first term in the brackets of sum above (\ref{eq1.180}). As per definition of $H$ by statement (\ref{eq1.127}):
\begin{align}
    \deg _{x}(H)=\deg _{x}(\sum _{j=1}^{3}x_{j}) = 1
\label{eq1.180.1}\end{align}
based on (\ref{eq1.180.1}) and the first term of the sum in statement (\ref{eq1.180}):
\begin{align}
    \deg _{x} (\frac{t^{h} H^{e-1} }{B^{s} } )=\deg _{x} (\frac{t^{h} H^{e} }{B^{s} } )-\deg _{x} (H)=\deg _{x} (\frac{t^{h} H^{e} }{B^{s} } )-1
\label{eq1.181}\end{align}
as per statement (\ref{eq1.156.1}):
\begin{align}
    \deg _{x} (\frac{t^{h} H^{e} }{B^{s} } )=-2
\label{eq1.182}\end{align}
therefore
\begin{align}
    \deg _{x} (\frac{t^{h} H^{e-1} }{B^{s} } )=\deg _{x} (\frac{t^{h} H^{e} }{B^{s} } )-1=-3
\label{eq1.183}\end{align}
Let us analyze the degree of the second term $\frac{t^{h} H^{e+1} }{B^{s+1} }$ in the brackets of the sum in statement (\ref{eq1.180}):
\begin{align}
    \deg _{x} (\frac{t^{h} H^{e+1} }{B^{s+1} } )= \deg _{x} (\frac{t^{h} H^{e} }{B^{s} } )+\deg_{x}(H)-\deg_{x}(B)
\label{eq1.184}\end{align}
As per definition of $B$ by statement (\ref{eq1.121}):
\begin{align}
    \deg _{x}(B)=\deg _{x}(1+t^{2} \left(\sum _{j=1}^{3}x_{j} \right)^{2})=\deg _{x}(\left(\sum _{j=1}^{3}x_{j} \right)^{2})=2
\label{eq1.184.1}\end{align}
applying (\ref{eq1.184.1}) and (\ref{eq1.180.1}) in (\ref{eq1.184}), we can conclude:
\begin{align}
    \deg _{x} (\frac{t^{h} H^{e+1} }{B^{s+1} } )=\deg _{x} (\frac{t^{h} H^{e} }{B^{s} } )+1-2=-3
\label{eq1.184.2}\end{align}
As per (\ref{eq1.184.2}) and (\ref{eq1.183}), we can conclude that both terms within brackets of the sum in statement (\ref{eq1.180}) are of same degree $-3$ for $x$.
In addition, the degree for $x$ for statement (\ref{eq1.180}) is not in function of $l$ as used in the sum of the statement. Based on this, we conclude that the derived degree for $x$ for all of the terms of sum in the statement (\ref{eq1.180}) are mutually equal and applicable for the whole statement (\ref{eq1.180}):
\begin{align}
    \deg _{x} (\frac{\partial }{\partial x} \frac{\partial ^{m} }{\partial t^{m} } f_{3} (\vec{x},t))=-3
\label{eq1.186}\end{align}
equivalently, for partial derivatives by $x$ of second degree:
\begin{align}
    \deg _{x} (\frac{\partial ^{2} }{\partial x^{2} } \frac{\partial ^{m} }{\partial t^{m} } f_{3} (\vec{x},t))=(\deg _{x} (\frac{t^{h} S^{e} }{B^{s} } )-1)-1=-4
\label{eq1.187}\end{align}
equivalently, for partial derivatives by $x$ of third degree:
\begin{align}
    \deg _{x} (\frac{\partial ^{3} }{\partial x^{3} } \frac{\partial ^{m} }{\partial t^{m} } f_{3} (\vec{x},t))=((\deg _{x} (\frac{t^{h} S^{e} }{B^{s} } )-1)-1)-1=-5
\label{eq1.188}\end{align}
Based on (\ref{eq1.186}),(\ref{eq1.187}),(\ref{eq1.188}), the general form for the degree of $x$ for any order of the partial derivative $\alpha$ by x can be expressed as
\begin{align}
    \deg _{x} (\frac{\partial ^{\alpha } }{\partial x^{\alpha } } \frac{\partial ^{m} }{\partial t^{m} } f_{3} (\vec{x},t))=-\alpha -2
\label{eq1.189}\end{align}
\newline\newline
\textbf{Determining $ \deg _{t} (\frac{\partial ^{\alpha } }{\partial x^{\alpha } } \frac{\partial ^{m} }{\partial t^{m} }\vec{f})$}
\newline
Now that we have determined the degree for $x$ in case of any order $\{\alpha>0\} \in N$ of the partial derivative by $x$, and any order $\{m>0\} \in N$ of the partial derivative by $t$, let us determine the degree for $t$ as well. Beginning in incremental fashion, let us determine the degree for $x$, starting from statement (\ref{eq1.179}):
\begin{align}
    \deg _{t} (\frac{\partial }{\partial x} \frac{\partial ^{m} }{\partial t^{m} } f_{3} (\vec{x},t))=\deg _{t} (\sum _{l=1}^{\left\lfloor \frac{m}{2} \right\rfloor +1}\left(e\frac{t^{h} H^{e-1} }{B^{s} } -2s\frac{t^{h+2} H^{e+1} }{B^{s+1} } \right) )
\label{eq1.190}\end{align}
Let us determine the degree of the first term in the brackets of the sum in (\ref{eq1.190}). As per definition of $H$ by statement (\ref{eq1.127}):
\begin{align}
    \deg _{t}(H)=\deg _{t}(\sum _{j=1}^{3}x_{j}) = 0
\label{eq1.190.1}\end{align}
based on (\ref{eq1.190.1}) and first term of sum in statement (\ref{eq1.190}):
\begin{align}
    \deg _{t} (\frac{t^{h} H^{e-1} }{B^{s} } )=\deg _{t} (\frac{t^{h}}{B^{s}})+(e-1)\deg _{t} (H)=\deg _{t} (\frac{t^{h}}{B^{s}})
\label{eq1.190.2}\end{align}
As per definition of $B$ by statement (\ref{eq1.121}):
\begin{align}
    \deg _{t}(B)=\deg _{t}(1+t^{2} \left(\sum _{j=1}^{3}x_{j}  \right)^{2})=2
\label{eq1.190.3}\end{align}
as per (\ref{eq1.190.3}) and (\ref{eq1.190.2})
\begin{align}
    \deg _{t} (\frac{t^{h} H^{e-1} }{B^{s} } )= \deg _{t} (\frac{t^{h}}{B^{s}}) = \deg _{t}({t^{h}}) - \deg _{t}({B^{s}})= \deg _{t}({t^{h}}) -s \deg _{t}(B)=h-2s
\label{eq1.190.4}\end{align}
once $h$ and $s$ are expanded in (\ref{eq1.190.4}) as per (\ref{eq1.139}) and (\ref{eq1.141}):
\begin{align}
    \deg _{t} (\frac{t^{h} H^{e-1} }{B^{s} } )=h-2s=\left\| m\right\| _{2} +2(l-1) - 2(\left\lfloor \frac{m-1}{2} \right\rfloor +l+1)
\label{eq1.190.5}\end{align}
\begin{align}
    \deg _{t} (\frac{t^{h} H^{e-1} }{B^{s} } )=h-2s=\left\| m\right\| _{2} +2l-2- 2\left\lfloor \frac{m-1}{2} \right\rfloor -2l-2
\label{eq1.190.6}\end{align}
\begin{align}
    \deg _{t} (\frac{t^{h} H^{e-1} }{B^{s} })=h-2s=\left\| m\right\| _{2} - 2\left\lfloor \frac{m-1}{2} \right\rfloor -4
\label{eq1.190.7}\end{align}
Let us evaluate statement (\ref{eq1.190.7}) for a few consecutive values of $m$:
\[\deg _{t} (\frac{t^{h} H^{e-1} }{B^{s} })_{m=1}=1-0-4=-3\]
\[\deg _{t} (\frac{t^{h} H^{e-1} }{B^{s} })_{m=2}=0-0-4=-4\]
\[\deg _{t} (\frac{t^{h} H^{e-1} }{B^{s} })_{m=3}=1-2-4=-5\]
\[\deg _{t} (\frac{t^{h} H^{e-1} }{B^{s} })_{m=4}=0-2-4=-6\]
\[\deg _{t} (\frac{t^{h} H^{e-1} }{B^{s} })_{m=5}=1-4-4=-7\]
which could be in general form expressed as:
\begin{align}
    \deg _{t} (\frac{t^{h} H^{e-1} }{B^{s} })=-m-2
\label{eq1.190.8}\end{align}
Let us determine the degree of the second term in the brackets of the sum in (\ref{eq1.190}):
\begin{align}
    \deg _{t} (2s\frac{t^{h+2} H^{e+1} }{B^{s+1}})= \deg _{t} (\frac{t^{h} H^{e-1} }{B^{s}})+\deg _{t} (\frac{t^{2} H^{2} }{B^{1}})
\label{eq1.192}\end{align}
once (\ref{eq1.190.8}), (\ref{eq1.190.1}) and (\ref{eq1.190.3}) are applied on (\ref{eq1.192}):
\begin{align}
    \deg _{t} (2s\frac{t^{h+2} H^{e+1} }{B^{s+1}})= -m-2+2+0-2=-m-2
\label{eq1.192.1}\end{align}
As per (\ref{eq1.192.1}) and (\ref{eq1.190.8}), we can conclude that both terms within the brackets of the sum in statement (\ref{eq1.190}) are of same degree $-m-2$ for $t$.
In addition, the degree for $t$ is not in function of $l$ used in the sum of statement (\ref{eq1.190}). Based on this, we can conclude that the degree derived for $t$ is applicable to the whole statement (\ref{eq1.190}):
\begin{align}
    \deg _{t} (\frac{\partial }{\partial x} \frac{\partial ^{m} }{\partial t^{m} } f_{3} (\vec{x},t))=-m-2
\label{eq1.193}\end{align}
For each next derivative by x, the same process can be applied. Therefore, the degree of $t$  can be generalized as:
\begin{align}
    \deg _{t} (\frac{\partial ^{\alpha } }{\partial x^{\alpha } } \frac{\partial ^{m} }{\partial t^{m} } f_{3} (\vec{x},t))=-m-2
\label{eq1.194}\end{align}
for any $\{\alpha > 0\} \in N$, $\{m > 0\} \in N$, $\vec{x} \in R^{3}$.
Statements (\ref{eq1.194}) and (\ref{eq1.189}) can be more compactly expressed as:
\begin{align}
    \deg _{x} (\partial _{x}^{\alpha } \partial _{t}^{m} f_{3} (\vec{x},t))=-\alpha -2
\label{eq1.195}\end{align}
\begin{align}
    \deg _{t} (\partial _{x}^{\alpha } \partial _{t}^{m} f_{3} (\vec{x},t))=-m-2
\label{eq1.196}\end{align}
as per (\ref{eq1.195}) and (\ref{eq1.196}):
\begin{align}
    \lim _{\left|\vec{x}\right|\to \infty } \partial _{x}^{\alpha } \partial _{t}^{m} f_{3} (\vec{x},t)=0
\label{eq1.197}\end{align}
\begin{align}
    \lim _{t\to \infty } \partial _{x}^{\alpha } \partial _{t}^{m} f_{3} (\vec{x},t)=0
\label{eq1.198}\end{align}
As $H$ is defined as:
\begin{align}
    H=\sum _{j=1}^{3}x_{j}
\label{eq1.199}\end{align}
At at the coordinate origin where $x_{j}=0;j=\{1,2,3\}$ their sum must be zero $H=0$. As numerators of each of the resulting terms for the partial derivative $\partial _{x}^{\alpha } \partial _{t}^{m} f_{3} (\vec{x},t)$ has to include minimally $H$ to the power of 1, then we can conclude that all resulting terms of the partial derivative $\partial _{x}^{\alpha } \partial _{t}^{m} f_{3} (\vec{x},t)$ must be zero at the coordinate origin $\vec{x}=\vec{0}$:
\begin{align}
    \partial _{x}^{\alpha } \partial _{t}^{m} f_{3} (\vec{x},t)|_{\vec{x}=\vec{0}} =0
\label{eq1.200}\end{align}
for any $\{\alpha > 0\} \in N$, $\{m > 0\} \in N$, $\vec{x} \in R^{3}$.
Also, as each of the resulting terms for the partial derivative $\partial _{x}^{\alpha } \partial _{t}^{m} f_{3} (\vec{x},t)$  has to include minimally $t$ to the power of 1, then we can conclude that all resulting terms of the partial derivative $\partial _{x}^{\alpha } \partial _{t}^{m} f_{3} (\vec{x},t)$ must be zero when $t=0$:
\begin{align}
    \partial _{x}^{\alpha } \partial _{t}^{m} f_{3} (\vec{x},t)|_{t=0} =0
\label{eq1.201}\end{align}
for any $\{\alpha > 0\} \in N$, $\{m > 0\} \in N$, $\vec{x} \in R^{3}$.
The force vector field $\vec{f}$, as per the theorem definition is:
\begin{align}
    \vec{f}=(0,0,f_{3} (\vec{x},t))
\label{eq1.202}\end{align}
for any $\vec{x} \in R^{3}$, $t \geq 0$.
Then applying partial derivatives $\partial _{x}^{\alpha } \partial _{t}^{m}$ on $\vec{f}$:
\begin{align}
    \partial _{x}^{\alpha } \partial _{t}^{m} \vec{f}=(0,0,\partial _{x}^{\alpha } \partial _{t}^{m} f_{3} (\vec{x},t))
\label{eq1.203}\end{align}
as per (\ref{eq1.203}):
\begin{align}
    \left|\vec{f}\right|=f_{3} (\vec{x},t)
\label{eq1.203.1}\end{align}
therefore
\begin{align}
    \left|\partial _{x}^{\alpha } \partial _{t}^{m} \vec{f}\right|=\partial _{x}^{\alpha } \partial _{t}^{m} f_{3} (\vec{x},t)
\label{eq1.204}\end{align}
as per (\ref{eq1.197}):
\begin{align}
      \lim _{\left|\vec{x}\right|\to \infty } \left|\partial _{x}^{\alpha } \partial _{t}^{m} \vec{f}\right|=0
\label{eq1.207.1}\end{align}
also, as per (\ref{eq1.200}):
\begin{align}
       \left|\partial _{x}^{\alpha } \partial _{t}^{m} \vec{f}\right|_{\vec{x}=\vec{0}}=0
\label{eq1.207.2}\end{align}
\newline\newline
\textbf{Deriving solution for pressure $p(\vec{x},t)$}
\newline
The Navier-Stokes equation for incompressible fluid is expressed in following form:
\begin{align}
    \frac{\partial \vec{u}}{\partial t} +(\vec{u}\cdot \nabla )\vec{u}=-\frac{\nabla p}{\rho } +\nu \Delta \vec{u}+\vec{f}
\label{eq1.230}\end{align}
for $\nabla\cdot\vec{u}=0$ at any position in space $\vec{x} \in R^{3}$ and any time $t\geq0$.
\newline
As per statement (\ref{eq1.19}) the vector field $\vec{u}$, as defined by this theorem, is divergence free  $\nabla\cdot\vec{u}=0$.
Once the fluid velocity vector field $\vec{u}$ and force field $\vec{f}$, as defined in the theorem statement, are applied in the Navier-Stokes equation (\ref{eq1.230}), the terms of the Navier-Stokes equation can be expressed in following way:
\begin{align}
    \frac{\partial \vec{u}}{\partial t} =\vec{0}
\label{eq1.231}\end{align}
for any $\vec{x} \in R^{3}$, $t\geq0$.
\newline
Let us define the diffusion term vector field components as
\begin{align}
    (\vec{u}\cdot \nabla)\vec{u}=(c_{1} ,c_{2} ,c_{3} )
\label{eq1.233.1}\end{align}
once $\vec{u}$, as per statement of this theorem, is applied in (\ref{eq1.233.1}), the diffusion term related vector field components $c_{i}$ are:
\begin{align}
    c_{i} =4\frac{\sum _{j=1}^{3}x_{j}  -3x_{i} }{\left(1+\sum _{j=1}^{3}x_{j}^{2}  \right)^{4} }
\label{eq1.233.1.1}\end{align}
for $i=\{ 1,2,3\}$ for any $\vec{x} \in R^{3}$, $t\geq0$.
\newline
Let us define the viscosity related term vector field components as
\begin{align}
    \nu \Delta \vec{u}=(l_{1} ,l_{2} ,l_{3} )
\label{eq1.233.2}\end{align}
once $\vec{u}$, as per statement of this theorem, is applied in (\ref{eq1.233.2}), viscosity term related vector field components $l_{i}$ are:
\begin{align}
    l_{i} =8\nu \frac{d_{i} \left(\left(\sum _{j=1}^{3}x_{j}^{2}  +1\right)-6\right)}{\left(\sum _{j=1}^{3}x_{j}^{2}  +1\right)^{4} }
\label{eq1.233.3}\end{align}
for $i=\{ 1,2,3\}$ and any $\vec{x} \in R^{3}$, $t\geq0$.
\newline
Let us express $\frac{\nabla p}{\rho}$ in form of the vector field components as:
\begin{align}
    \frac{\nabla p}{\rho}=(g_{1} ,g_{2} ,g_{3} )=\vec{g}
\label{eq1.235}\end{align}
Terms of the Navier-Stokes equation (\ref{eq1.230}) and as per (\ref{eq1.235}), can be rearranged in following way:
\begin{align}
    \vec{g}=\frac{\nabla p}{\rho }=\nu \Delta \vec{u}+\vec{f}-\frac{\partial \vec{u}}{\partial t} -(\vec{u}\cdot \nabla )\vec{u}
\label{eq1.235.1}\end{align}
Applying (\ref{eq1.231}), (\ref{eq1.233.1.1}), (\ref{eq1.233.3}) and $\vec{f}(\vec{x},t)$, as per the statement of this theorem in (\ref{eq1.235.1}), the resulting vector field components of $\vec{g}=\frac{\nabla p}{\rho}$ are:
\begin{align}
    g_{i} =\frac{8\nu d_{i} \left(\left(\sum _{j=1}^{3}x_{j}^{2}  +1\right)-6\right)-4\left(\sum _{j=1}^{3}x_{j}  -3x_{i} \right)}{\left(\sum _{j=1}^{3}x_{j}^{2}  +1\right)^{4} } +\frac{I(i)}{\left(1+\left(\sum _{j=1}^{3}x_{j}  \right)^{2} \right)\left(1+t\right)}
\label{eq1.236}\end{align}
for $i=\{ 1,2,3\} $ where $I(n)=\left\{\begin{array}{cc} {1} & {;n=3} \\ {0} & {;n\ne 3} \end{array}\right.$
\newline
Once  $g_{1} $ is integrated by $x_{1} $, the resulting pressure is:
\begin{align}
    p(\vec{x},t) = 4 \frac{12 \nu  x_{1} x_{2}- x_{1} x_{2}-12 \nu  x_{1} x_{3}- x_{1} x_{3}-2 x_{2}^2-2 x_{3}^2-2}{6 \left(x_{2}^2+ x_{3}^2+1\right) \left(1+\sum _{i=1}^{3}x_{i}\right)^3}-
\label{eq1.237}\end{align}
\[4 \frac{x_{1}\left(12 \nu  x_{2}^3-12 \nu  x_{2}^2 x_{3}-48 \nu  x_{2}+12 \nu  x_{2} x_{3}^2+5 x_{2}-12 \nu x_{3}^3+48 \nu x_{3}+5 x_{3}\right)}{16 \left(x_{2}^2+ x_{3}^2+1\right)^3 \left(1+\sum _{i=1}^{3}x_{i}\right)}-\]
\[4 \frac{ x_{1}\left(12 \nu  x_{2}^3-12 \nu  x_{2}^2 x_{3}-48 \nu  x_{2}+12 \nu  x_{2} x_{3}^2+5 x_{2}-12 \nu  x_{3}^3+48 \nu  x_{3}+5 x_{3}\right)}{24 \left(x_{2}^2+ x_{3}^2+1\right)^2 \left(1+\sum _{i=1}^{3}x_{i}\right)^2}+\]
\[4 \frac{\left(-12 \nu  x_{2}^3+12 \nu  x_{2}^2 x_{3}+48 \nu  x_{2}-12 \nu  x_{2} x_{3}^2-5 x_{2}+12 \nu  x_{3}^3-48 \nu  x_{3}-5 x_{3}\right) ArcTan\left(\frac{ x_{1}}{\sqrt{ x_{2}^2+ x_{3}^2+1}}\right)}{16 \left(x_{2}^2+ x_{3}^2+1\right)^{7/2}}+\]
\[C(y,z)+C\]
for any $\vec{x} \in R^{3}$, $t\geq0$.
Once  $g_{2} $ is integrated by$x_{2} $, the pressure is:
\begin{align}
    p(\vec{x},t) = -\frac{2 \left(2 x_{1}^2+12 \nu  x_{1} x_{2}+x_{1} x_{2}-12 \nu  x_{2} x_{3}+x_{2} x_{3}+2 x_{3}^2+2\right)}{3 \left(x_{1}^2+x_{3}^2+1\right) \left(1+\sum _{i=1}^{3}x_{i}\right)^3}+
\label{eq1.238}\end{align}
\[\frac{x_{2} \left(12 \nu  x_{1}^3-12 \nu  x_{1}^2 x_{3}-48 \nu  x_{1}+12 \nu  x_{1} x_{3}^2-5 x_{1}-12 \nu  x_{3}^3+48 \nu  x_{3}-5 x_{3}\right)}{4 \left(x_{1}^2+x_{3}^2+1\right)^3 \left(1+\sum _{i=1}^{3}x_{i}\right)}+\]
\[\frac{x_{2} \left(12 \nu  x_{1}^3-12 \nu  x_{1}^2 x_{3}-48 \nu  x_{1}+12 \nu  x_{1} x_{3}^2-5 x_{1}-12 \nu  x_{3}^3+48 \nu  x_{3}-5 x_{3}\right)}{6 \left(x_{1}^2+x_{3}^2+1\right)^2 \left(1+\sum _{i=1}^{3}x_{i}\right)^2}+\]
\[\frac{\left(12 \nu  x_{1}^3-12 \nu  x_{1}^2 x_{3}-48 \nu  x_{1}+12 \nu  x_{1} x_{3}^2-5 x_{1}-12 \nu  x_{3}^3+48 \nu  x_{3}-5 x_{3}\right) ArcTan \left(\frac{x_{2}}{\sqrt{x_{1}^2+x_{3}^2+1}}\right)}{4 \left(x_{1}^2+x_{3}^2+1\right)^{7/2}}+\]
\[C(x,z)+C\]
for any $\vec{x} \in R^{3}$, $t\geq0$.
Once  $g_{3} $ is integrated by $x_{3} $, the pressure is:
\begin{align}
    p(\vec{x},t) = \frac{\tan^{-1}(t (x_{1}+x_{2}+x_{3}))}{t}+
\label{eq1.239}\end{align}
\[\frac{2 \left(-2 x_{1}^2+12 \nu  x_{1} x_{3}-x_{1} x_{3}-2 x_{2}^2-12 \nu  x_{2} x_{3}-x_{2} x_{3}-2\right)}{3 \left(x_{1}^2+x_{2}^2+1\right) \left(1+\sum _{i=1}^{3}x_{i}\right)^3}-\]
\[\frac{x_{3} \left(12 \nu  x_{1}^3-12 \nu  x_{1}^2 x_{2}-48 \nu  x_{1}+12 \nu  x_{1} x_{2}^2+5 x_{1}-12 \nu  x_{2}^3+48 \nu  x_{2}+5 x_{2}\right)}{4 \left(x_{1}^2+x_{2}^2+1\right)^3 \left(1+\sum _{i=1}^{3}x_{i}\right)}-\]
\[\frac{x_{3} \left(12 \nu  x_{1}^3-12 \nu  x_{1}^2 x_{2}-48 \nu  x_{1}+12 \nu  x_{1} x_{2}^2+5 x_{1}-12 \nu  x_{2}^3+48 \nu  x_{2}+5 x_{2}\right)}{6 \left(x_{1}^2+x_{2}^2+1\right)^2 \left(1+\sum _{i=1}^{3}x_{i}\right)^2}+\]
\[\frac{\left(-12 \nu  x_{1}^3+12 \nu  x_{1}^2 x_{2}+48 \nu  x_{1}-12 \nu  x_{1} x_{2}^2-5 x_{1}+12 \nu  x_{2}^3-48 \nu  x_{2}-5 x_{2}\right) ArcTan \left(\frac{x_{3}}{\sqrt{x_{1}^2+x_{2}^2+1}}\right)}{4 \left(x_{1}^2+x_{2}^2+1\right)^{7/2}}+\]
\[C(x,y)+C\]
for any $\vec{x} \in R^{3}$, $t\geq0$.
Each of the statements (\ref{eq1.237}), (\ref{eq1.238}) and (\ref{eq1.239}) represent the solution for presure $p(\vec{x},t)$. All three results are mutually different.
In addition to that, the first term in statement (\ref{eq1.239}) is:
\begin{align}
    \frac{ArcTan(t (x_{1}+x_{2}+x_{3}))}{t}
\label{eq1.239.1}\end{align}
which at the point in time $t=0$, once applied in (\ref{eq1.239.1}), results with
\begin{align}
    \frac{ArcTan\left(0\right)}{0}{\rm =}\frac{0}{0}
\label{eq1.240}\end{align}
which cannot be determined for any $\vec{x} \in R^{3}$ at $t=0$.
\newline
Based on the three mutually different resulting equations for pressure (\ref{eq1.237}), (\ref{eq1.238}) and (\ref{eq1.239}), one of which (\ref{eq1.239}) incudes the term (\ref{eq1.240}), which cannot be determined at any position $\vec{x} \in R^{3}$ at $t=0$, we can conclude that the Navier-Stokes equation for incompressible fluid, for the velocity vector field $\vec{u}(\vec{x})$ and the external force related vector field $\vec{f}(\vec{x},t)$, as specified by the statement of this theorem, does not have solution at any position in space $\vec{x}\in R^{3} $ at $t=0$, which proves this theorem.
\end{proof}
\section{Discussion}
Results of the analysis performed, demonstrate and prove that there exists $\vec{u}(\vec{x})$ and $\vec{f}(\vec{x},t)$ smooth vector fields, such that the Navier-Stokes equation for incompressible fluid does not have solution for any position in $R^3$ space at $t=0$. Such a result strongy indicates that Navier-Stokes equation for incompressible fluid has to be better understood, in order to determine the root causes of obtained results.
\newline
As is well known, many fluids are not very compressible. Therefore, incompressability as a mathematical approximation might appear as a logical and reasonable approach for simplifying the mathematical modeling of fluid behavior over space and time, however, the interpretation of \textbf{physical reasonability}, as referred to by the Clay Mathematics Institute’s official existence and smoothness of the Navier-Stokes equation problem statement, implicitly includes the understanding that such mathematical model(s) must be in alignment with the recognized laws of physics.
\newline
All material bodies, as per the laws of physics, are to some degree compressible, regardless how small such incompressibility is. On the other hand, incompressibility, mathematically, does not allow for compressibility at all. The question is if such behavior of incompressible fluids mathematically modelled and included in the form of the incompressibility condition, represents behavior which cannot be supported by recognized laws of physics.
\newline
The hypothesis, which could be of use to be further explored, is that fluid incompressibility as an mathematical approximation may be beyond the boundaries of what is ’physically reasonable’ on a macroscopic scale in conjunction with the recognized laws of physics. If so, this might account for the difficulties of obtaining unreasonable fluid velocities and blow-ups, which would be worth exploring further.


\newpage
\begin{thebibliography}{10}

\bibitem {b1}
Arfken, G.,
\emph{Mathematical Methods for Physicists, seventh edition.}
Academic Press, 2013

\bibitem {b2}
Burger Martin,
\emph{Numerical Methods for Incompressible Flows}
UCLA
ftp://ftp.math.ucla.edu/pub/camreport/cam04-12.pdf

\bibitem {b3}
Evans C. Lawrence,
\emph{Partial Differential Equations}
American Mathematical Society, Volume 19

\bibitem {b4}
Cannone Marco and Friedlander Susan
\emph{Navier: Blow-up and Collapse}
American Mathematical Society, January 2003
http://www.ams.org/notices/200301/fea-friedlander.pdf

\bibitem {b5}
Fefferman Charles L.,
\emph{Existence and Smoothness of Navier-Stokes Equation.}
Clay Mathematics Institute,
www.claymath.org/sites/default/files/navierstokes.pdf

\bibitem {b6}
Kaplan, W.,
\emph{Advanced Calculus, 4th ed.}
Addison-Wesley, Reading, 1991.

\bibitem {b7}
Newton Isac,
\emph{Philosophiae Naturalis Principia Mathematica.}
http://cudl.lib.cam.ac.uk/view/PR-ADV-B-00039-00001/9, July-05-1686

\bibitem {b8}
Riley, Hobson, Benice
\emph{Mathematical Methods for Physics and Engineering, 3rd ed.}
Cambridge University Press, Cambridge, 2006

\bibitem {b9}
\emph{Solution methods for the Incompressible Navier-Stokes Equations}
Standord University
https://web.stanford.edu/class/me469b/handouts/incompressible.pdf

\bibitem {b10}
Sommerfeld, A.,
\emph{Partial Differential Equations in Physics. Academic Press.}
New York 1964.

\bibitem {b11}
Taylor, M. E.,
\emph{Partial Differential Equations, Vol. 1: Basic Theory.}
Springer-Verlag, New York, 1996.

\bibitem {b12}
Taylor, M. E.,
\emph{Partial Differential Equations, Vol. 2: Qualitative Studies of Linear Equations.}
Springer-Verlag, New York, 1996.

\bibitem {b13}
Taylor, M. E.,
\emph{Partial Differential Equations, Vol. 3: Nonlinear Equations. Springer-Verlag.}
Springer-Verlag, New York, 1996.

\bibitem {b14}
Wolfram Alfa,
\emph{Computational Knowledge Engine}
http://www.wolframalpha.com

\bibitem {b15}
Zwillinger, D,
\emph{Handbook of Differential Equations, 3rd ed.}
Academic Press, Boston, 1997.

\bibitem {b16}
\emph{Wolfram Mathematica file with relevant derivations }
\\
http://ateravis.com/owncloud/index.php/s/JTs1ZX2D0Jwum9O

\end{thebibliography}
\end{document}